\documentclass[journal]{IEEEtran}
\usepackage{graphicx}

\usepackage{caption}

\usepackage{amsthm}
\usepackage{cite}

 \usepackage{multirow} 

\usepackage[utf8]{inputenc}

\ifCLASSINFOpdf
\else
\fi
\usepackage{amsmath}
\usepackage{xcolor}

\usepackage{algorithmic}
\usepackage{algorithm}

\usepackage{bm}
\usepackage[caption=false,font=footnotesize]{subfig}
\usepackage{graphicx}
\usepackage{array}
\usepackage{fixltx2e}

\usepackage{amssymb}

\begin{document}
%

\title{Joint Channel Estimation, Detection, and \\Resource Allocation for OTFS-RSMA
}
%
%
%

\author{Chaedam~Son and Si-Hyeon Lee,~\IEEEmembership{Senior Member,~IEEE}
\thanks{} 
\thanks{The authors  are with the School of Electrical Engineering, Korea Advanced Institute of Science and Technology (KAIST), Daejeon 34141, South Korea (e-mail: \{scd5929,sihyeon\}@kaist.ac.kr).  
(Corresponding author: Si-Hyeon Lee.)}
}

\allowdisplaybreaks[2]

\newcommand{\cd}[1]{\textcolor{blue}{{#1}}}
\newcommand{\sh}[1]{\textcolor{red}{{#1}}}
\newcommand{\re}[1]{\textcolor{green}{{#1}}}


\maketitle
\newtheorem{theorem}{Theorem}
\newtheorem{lemma}{Lemma}
\newtheorem{remark}{Remark}
\begin{abstract}
Orthogonal time frequency space (OTFS) modulation is well suited for high-mobility communications, but channel estimation and efficient
multi-user transmission remain challenging.
In this paper, we develop an OTFS-rate splitting multiple access (RSMA)
framework that jointly considers pilot transmission, channel estimation,
message detection, and resource allocation under imperfect channel state
information (CSI).
We consider guard-based and superimposed transmission for the common and
private messages to account for the tradeoff between channel-estimation
reliability and spectral efficiency (SE).
We derive a linear minimum mean square error (LMMSE) channel estimator,
error-aware message-passing detectors, and tractable surrogate rate
expressions that account for channel-estimation errors.
Based on these expressions, we formulate a joint resource-allocation
problem and develop an SCA-based two-dimensional search algorithm together
with a low-complexity alternative that exploits the analytical structure
of the resource allocation.
Numerical results demonstrate reliable channel-estimation and BER
performance and show that configurations with guard-based common-message
transmission achieve favorable sum-SE performance, while the preferred
private-message structure depends on the operating regime.
The low-complexity algorithm achieves near-SCA performance with substantially
reduced computation time, and OTFS-RSMA remains competitive with OTFS-NOMA
while requiring only a single SIC stage.
\end{abstract}


%
\IEEEpeerreviewmaketitle

\section{Introduction}
\IEEEPARstart{A}{s} sixth-generation (6G) wireless communications evolve, high-mobility applications such as autonomous driving, unmanned aerial vehicles (UAVs), and low Earth orbit (LEO) satellite communications are becoming increasingly important. The resulting severe Doppler effects give rise to doubly selective channels, where conventional cyclic prefix
(CP)-aided orthogonal frequency division multiplexing (OFDM) suffers from inter-carrier interference (ICI). To address this challenge, orthogonal time frequency space (OTFS) modulation has emerged as a promising
waveform~\cite{OTFS_review1,OTFS_review2}. By mapping information symbols onto the delay-Doppler domain, OTFS exploits the sparse channel structure
to provide a nearly time-invariant effective channel, reduced pilot
overhead, and robust data detection~\cite{OTFS-NOMA}.

However, in multi-user OTFS systems, orthogonal resource allocation
reduces the delay--Doppler resources available to each user, potentially
limiting spectral efficiency~\cite{coexistence}.
This has motivated non-orthogonal resource sharing, such as
non-orthogonal multiple access (NOMA), which employs superposition and
successive interference cancellation (SIC)~\cite{OTFS-NOMA}.
While enabling efficient resource sharing, NOMA relies on the SIC decoding
order and may require multiple SIC stages.
More recently, rate-splitting multiple access (RSMA) has been introduced
into OTFS systems~\cite{OTFS-RSMA1,OTFS-RSMA2,OTFS-RSMA3, OTFS-RSMA_LEO}.
By splitting user messages into common and private parts, RSMA flexibly
manages multi-user interference while reducing the SIC burden.
Existing studies have investigated OTFS-RSMA for sum-rate
maximization~\cite{OTFS-RSMA1}, heterogeneous user
mobility~\cite{OTFS-RSMA2}, deep learning-based
detection~\cite{OTFS-RSMA3}, and LEO
satellite communications~\cite{OTFS-RSMA_LEO}.  
However, these studies generally assume that the required channel state
information is available and do not explicitly address channel estimation in
the OTFS-RSMA framework. 
Explicitly accounting for channel estimation is particularly important in
OTFS-RSMA because the pilot can experience interference from both the common and
private messages depending on their transmission
structures, while estimation errors can propagate through SIC. 
Existing studies on OTFS channel estimation and pilot
design~\cite{OTFS_channel_est1,OTFS_channel_est2,OTFS_superimposed0,
OTFS_superimposed1,OTFS_superimposed2,OTFS_superimposed3}
have mainly considered guard-based and superimposed pilot schemes.
Guard-based schemes protect the pilot from data interference, enabling
reliable channel estimation at the cost of guard
overhead~\cite{OTFS_channel_est1,OTFS_channel_est2}.
In contrast, superimposed schemes improve spectral efficiency (SE) by
allowing pilots and data to share the same resources, but data interference
can degrade channel estimation~\cite{OTFS_superimposed0,
OTFS_superimposed1,OTFS_superimposed2,OTFS_superimposed3}.
These approaches thus exhibit a fundamental tradeoff between channel
estimation reliability and SE. {In OTFS-RSMA, this tradeoff further interacts with the transmission and decoding of the common and private messages, motivating a joint system design that explicitly accounts for channel estimation.}

In this paper, we develop an OTFS-RSMA framework that jointly considers
pilot transmission, channel estimation, message detection, and resource
allocation under imperfect CSI. 
We consider guard-based and superimposed transmission for both the common
and private messages, allowing their effects on channel estimation and
data transmission to be incorporated into the system design.
We develop channel-estimation and SIC-based message-detection methods and
characterize the resulting channel-estimation errors. Based on this error characterization, we formulate a joint resource-allocation
problem over the pilot, common-data, and private-message powers and the
common-rate shares among users. 
The resulting formulation explicitly captures the coupling between resource
allocation and channel-estimation errors. 
We further investigate the resulting guard-based and superimposed
configurations and compare their sum-SE performance.
To the best of our knowledge, such joint consideration of pilot transmission, channel
estimation, message detection, and resource allocation has not been
addressed in existing OTFS-RSMA studies.

{The main contributions   are summarized as follows.
\begin{itemize}
    \item \textbf{Integrated OTFS-RSMA framework:}
    We develop an OTFS-RSMA framework that jointly considers pilot
    transmission, channel estimation, message detection, and resource
    allocation under imperfect CSI.
    Guard-based and superimposed transmission are considered for both
    the common and private messages, enabling their effects on channel
    estimation and data transmission to be incorporated into the overall
    system design.
    \item \textbf{Channel estimation and error-aware message detection:}
    We derive closed-form LMMSE channel estimates and the corresponding
    mean square error (MSE) expressions in the presence of data interference.
    Based on the resulting channel-error characterization, we develop
    message-passing detectors for the common and private messages that
    explicitly account for channel-estimation errors and residual
    interference through SIC.
    \item \textbf{Resource allocation:}
    We derive tractable surrogate rate expressions that account for
    channel-estimation errors and formulate a joint resource-allocation
    problem over the pilot, common-data, and private-message powers and
    the common-rate allocation.
    We develop a successive convex approximation (SCA)-based 2D search algorithm and a low-complexity
    alternative based on a conservative private-rate approximation.
    For the latter, we further characterize the resource-allocation
    structure and derive closed-form common-rate and private-power
    allocation rules.
    \item \textbf{Numerical evaluation:}
    Numerical results demonstrate the channel-estimation and  bit error rate (BER)
    performance of the proposed framework and reveal the sum-SE tradeoff
    among the considered guard-based and superimposed configurations.
    In particular, configurations employing guard-based transmission for
    the common message achieve favorable sum-SE performance, while the
    preferred private-message structure depends on the operating regime.
    The low-complexity resource-allocation algorithm achieves near-SCA
    performance with substantially reduced computation time, and
    OTFS-RSMA remains competitive with OTFS-NOMA while requiring only a
    single SIC stage.
\end{itemize}}
\textit{Notation}: Scalars are regular italic, whereas vectors and matrices are boldface; their dimensions are specified at first definition (e.g., $x$, $\pmb{x}$, $\pmb{X}$). The $(i,j)$-th element of a matrix $\pmb{X}$ is $\pmb{X}(i,j)$, and the $i$-th element of a vector $\pmb{a}$ is $\pmb{a}(i)$. For a matrix $\pmb{A}$, $\pmb{A}^T$ and $\pmb{A}^H$ denote the transpose and Hermitian transpose, respectively, and $\text{Tr}(\pmb{A})$ denotes the trace. For a scalar $x$, $|x|$ is the absolute value, and $\|\cdot\|$ denotes the Euclidean norm for vectors or the Frobenius norm for matrices. The vectorization operator $\text{vec}(\pmb{A})$ stacks the columns of $\pmb{A}$ into a single vector. The circularly symmetric complex Gaussian distribution with mean $\pmb{\mu}$ and covariance $\pmb{\Sigma}$ is denoted by $\mathcal{\mathcal{CN}}(\pmb{\mu},\pmb{\Sigma})$. The expectation and variance operators are denoted by $\mathbb{E}[\cdot]$ and $\mathbb{V}[\cdot]$, respectively. 
\section{System Model}
We consider a downlink multi-user RSMA system based on OTFS modulation, where
a single-antenna base station (BS) serves \(U\) single-antenna users indexed by
\(\mathcal{U}\triangleq\{1,\ldots,U\}\).
Each OTFS frame has duration \(T_f = NT\) and bandwidth \(B = M\Delta f\),
where \(N\) and \(M\) denote the numbers of time and frequency samples,
respectively, and \(T\) and \(\Delta f=1/T\) denote the sampling intervals
along the time and frequency axes, respectively.
We assume that both the BS and the users have statistical CSI~\cite{OTFS-statisticalCSI}.
The BS aims to deliver an independent message to each user. Following the RSMA framework, each user’s message is split into a common part and a private part. The common parts of all users are combined into a single common message that is decoded by all users, whereas each private part is encoded into a private message intended only for its corresponding user. Each user first decodes the common message and removes it via SIC, and then decodes its own private message \cite{RSMA_power}. 

In OTFS systems, pilot transmission schemes can be broadly classified into guard-based and superimposed approaches. The guard-based pilot scheme employs guard symbols to eliminate interference between pilot and data symbols, thereby enabling accurate channel estimation at the expense of reduced SE \cite{OTFS_channel_est1}. In contrast, the superimposed pilot scheme improves resource utilization by removing guard symbols; however, it introduces pilot--data interference and typically requires SIC for reliable detection \cite{OTFS_superimposed1,OTFS_superimposed2,OTFS_superimposed3}.
{We consider the possible pilot transmission configurations for OTFS-RSMA
by allowing the common and private messages to independently adopt either a guard-based (G) or superimposed (S) structure.
This results in four configurations, denoted by GG, GS, SG, and SS, where the first and second letters indicate the structures of the common and private messages, respectively.
For example, GS employs a guard-based common message and superimposed private messages, whereas SG employs a superimposed common message and guard-based private messages; GG and SS apply the same guard-based and superimposed structures, respectively, to both message types.}

In the following, we focus on the asymmetric GS configuration, which employs guard-based transmission for the common message and superimposed transmission for the private messages. Based on this configuration,
we present the system model and develop the corresponding channel-estimation,
detection, and resource-allocation methods. 
The extensions to the other three configurations are briefly discussed
later.
In the GS configuration, the pilot is embedded in the common message and
protected from the common data by a guard region, while the private messages
are transmitted over the entire delay-Doppler (DD) grid without guard regions, including
the pilot region.
Each user first estimates its channel using the pilot, decodes the common
message, removes it via SIC, and then decodes its own private message.
The detailed channel and signal models are provided in the following
subsections.


\subsection{Channel Model}
The delay--Doppler channel $h_u(\tau,\nu)$ of the $u$-th user is given as \(
     h_u(\tau, \nu)= \sum_{q=1}^{Q_u} h_{u,q}\,\delta\big(\tau - \tau_{u,q}\big)\,\delta\big(\nu - \nu_{u,q}\big),\)
where $Q_u$ is the number of propagation paths of user $u$ and $h_{u,q}$, $\tau_{u,q}$, and $\nu_{u,q}$ are the complex channel gain, delay, and Doppler shift for the $q$-th path of user $u$, respectively. For analytical tractability, we assume that $Q_u$ is identical for all $u$, i.e., $Q_1 = Q_2 = \cdots = Q_U=Q $.
The delay and Doppler shifts of the $q$-th path of user $u$ are given by
$\tau_{u,q} = \frac{l_{u,q}}{M \Delta f}$ and $ \nu_{u,q} = \frac{k_{u,q}}{N T}$, respectively, where $l_{u,q}$ and $k_{u,q}$ denote the delay and Doppler indices and are assumed to be bounded \cite{OTFS_channel_est1}. 
We denote the maximum delay and Doppler indices over all users and paths by $l_{\text{max}} \triangleq \max_{u,q} l_{u,q}$ and $k_{\text{max}} \triangleq \max_{u,q} |k_{u,q}|$, respectively. Moreover, we assume that $M$ and $N$ are sufficiently large so that the fractional delay and Doppler components can be neglected, i.e., $l_{u,q}$ and $k_{u,q}$ are assumed to be integers \cite{OTFS_superimposed4,OTFS_superimposed0}. Furthermore, we assume that the delay--Doppler domain channel gains $h_{u,q}$ are uncorrelated across paths and users, and follow a complex Gaussian distribution\cite{OTFS_superimposed1}. {Specifically, the channel gain vector of user $u$ is given by
$\pmb{h}_u = [h_{u,1}, \ldots, h_{u,Q}]^{\mathsf{T}}
       \sim \mathcal{\mathcal{CN}}\big(\pmb{0}_Q, \pmb{K}_{\pmb{h}_u, \pmb{h}_u}\big),
$
with covariance matrix {$\pmb{K}_{\pmb{h}_u, \pmb{h}_u}= \text{diag}\big(\sigma_{u,1}^2, \sigma_{u,2}^2, \ldots, \sigma_{u,Q}^2\big)$}. We define user $u$'s effective channel gain as $\sigma_{u}^2=\sum_q \sigma_{u,q}^2$}.

\subsection{Signal Model}
We consider a single common message and $U$ private messages for the $U$ users in the delay--Doppler domain. Let $x_{\text{c}}[l,k] $ and $ x_{\text{p},u}[l,k] $ denote the common-message symbol and the private-message symbol for user $u$, respectively, at the $(l,k)$-th delay--Doppler index, where $ l \in \{0,\dots,M-1\} $ and $ k \in \{0,\dots,N-1\} $. The common message has an asymmetric structure including pilot, guard, and data symbols, following the guard-based pilot structure \cite{OTFS_channel_est1}, whereas each private message consists only of data symbols. {Let \(\mathcal{A}_{\text{pilot}}^{\text{DD}}=\{(l_{\text{r}},k_{\text{r}})\}\), \(\mathcal{A}_{\text{guard}}^{\text{DD}}=\{(l,k): l_{\text{r}} - l_{\text{max}} \le l \le l_{\text{r}} + l_{\text{max}}, k_{\text{r}} - 2k_{\text{max}} \le k \le k_{\text{r}} + 2k_{\text{max}}\}\setminus \mathcal{A}_{\text{pilot}}^{\text{DD}}\), and \(\mathcal{A}_{\text{data}}^{\text{DD}}=\{(l,k):  0\leq l \leq M-1, 0\leq k \leq N-1\}\setminus \mathcal{A}_{\text{pilot}}^{\text{DD}} \setminus \mathcal{A}_{\text{guard}}^{\text{DD}} \)} {for some  $0\leq l_{\text{r}} \leq M-1$ and $0\leq k_{\text{r}} \leq N-1$} denote the sets of delay--Doppler indices corresponding to pilot, guard, and data positions of the common message, respectively. The common message symbol at the $(l,k)$-th delay-Doppler index is constructed as
\begin{align}
    x_{\text{c}}[l,k] =
    \begin{cases}
        x_{\text{c},\text{r}}, 
        & (l,k)\in\mathcal{A}_{\text{pilot}}^{\text{DD}}, \\[2pt]
        0, 
        &(l,k)\in\mathcal{A}_{\text{guard}}^{\text{DD}}, \\[2pt]
        x_{\text{c,d}}[l,k], 
        & (l,k)\in\mathcal{A}_{\text{data}}^{\text{DD}},
    \end{cases}
\end{align}
where $x_{\text{c,r}}$ {is the pilot (reference)} symbol~\cite{OTFS_channel_est1} and the symbols $x_{\text{c,d}}[l,k]$ are the common data symbols. Each data symbol is drawn from a finite constellation set $\mathcal{S}$ of cardinality $S$, i.e., $x_{\text{c,d}}[l,k]\in\mathcal{S}$ and $x_{\text{p},u}[l,k]\in\mathcal{S}$. Let $\mathcal{S}_{l,k}$ denote the constellation set from which $x_{\text{c}}[l,k]$ is drawn, i.e.,  
\[
\mathcal{S}_{l,k}=
\begin{cases}
\{x_{\text{c,r}}\}, & (l,k)\in\mathcal{A}_{\text{pilot}}^{\text{DD}},\\
\{0\}, & (l,k)\in\mathcal{A}_{\text{guard}}^{\text{DD}},\\
\mathcal{S}, & (l,k)\in\mathcal{A}_{\text{data}}^{\text{DD}}. 
\end{cases}
\]
We assume that the data symbols are independent, zero-mean, and uncorrelated across delay--Doppler indices and users, with
$\mathbb{E}[x_{\text{c,d}}[l,k]]=0$,
$\mathbb{E}[x_{\text{p},u}[l,k]]=0$,
$\mathbb{E}[x_{\text{c,d}}[l,k]x_{\text{c,d}}^{*}[l',k']]=0$, and
$\mathbb{E}[x_{\text{p},u}[l,k]x_{\text{p},v}^{*}[l',k']]=0$ for $(l,k)\neq(l',k')$ or $u\neq v$. {In addition, let $P_{\text{c,d}}$, $P_{\text{p},u}$ and $P_{\text{c,r}}$ denote the transmit powers of the common data symbols, private-message symbols for user $u$, and the pilot symbol, respectively, i.e., 
$\mathbb{E}[|x_{\text{c,d}}[l,k]|^2]=P_{\text{c,d}}$,
$\mathbb{E}[|x_{\text{p},u}[l,k]|^2]=P_{\text{p},u}$, and $|x_{\text{c,r}}|^2=P_{\text{c,r}}$.} We further define \(P_{\text{p}}=\sum_{u=1}^{U}P_{\text{p},u}\) as the total private-message power and \(N_g=(4k_{\text{max}}+1)(2l_{\text{max}}+1)-1\) as the number of guard symbols.

Accordingly, the total transmit symbol at the $(l,k)$-th delay--Doppler index is given by
\begin{align}
    x[l,k] = x_{\text{c}}[l,k] + \sum_{u=1}^{U} x_{\text{p},u}[l,k].
\end{align}
Let $\pmb{X} \in \mathbb{C}^{M \times N}$ denote the matrix formed by the symbols $\{x[l,k]\}$, where $\pmb{X}(l+1,k+1) = x[l,k]$. Similarly, let $\pmb{X}_{\text{c}}$ and $\pmb{X}_{\text{p},u}$ denote the matrices formed by $\{x_{\text{c}}[l,k]\}$ and $\{x_{\text{p},u}[l,k]\}$, respectively. Then, we have
\begin{align}
    \pmb{X} = \pmb{X}_{\text{c}} + \sum_{u=1}^{U} \pmb{X}_{\text{p},u}.
\end{align}
The delay--Doppler symbols are mapped to the time--frequency domain via the inverse symplectic finite Fourier transform (ISFFT) and transmitted over the wireless channel by a rectangular transmit waveform\cite{OTFS_superimposed1}. After standard OTFS demodulation, the received signal can be represented in the delay--Doppler domain as follows \cite{OTFS_superimposed4}: 
\begin{align}
    \!y_u[l,k]
    \!=\!\sum_{q=1}^{Q}\!h_{u,q}\alpha_{u,q}(l,k)
      x([l-l_{u,q}]_M\!,\![k-k_{u,q}]_N)\!+\! n_u[l,k],
      \notag
\end{align}
where $n_u[l,k]\sim \mathcal{\mathcal{CN}}(0,\sigma_n^2)$ is the additive white Gaussian noise (AWGN) and $ \alpha_{u,q}(l,k) $ denotes the additional phase shift induced by the rectangular waveform, which is given by
\begin{align}
\alpha_{u,q}(l,k)=
\begin{cases}
e^{-j2\pi \frac{k}{N}} z^{\,k_{u,q}\big((l-l_{u,q})\big)_M}, & l<l_{u,q},\\[4pt]
z^{\,k_{u,q}\big((l-l_{u,q})\big)_M}, & l\ge l_{u,q},
\end{cases}
\end{align}
with $ z = e^{j2\pi/(MN)} $.

For analytical convenience, we vectorize the received delay--Doppler symbols into 
\(\pmb{y}_u \in \mathbb{C}^{MN}\). The vectorized input-output relation is given by
\begin{align}
\label{receive signal}
    \pmb{y}_u
    = \pmb{H}_u \left( \pmb{x}_{\text{c}} + \sum_{j=1}^{U} \pmb{x}_{\text{p},j} \right) + \pmb{n}_u,
\end{align}
where \(\pmb{H}_u \in\mathbb{C}^{MN\times MN}\) is the equivalent delay--Doppler-domain channel matrix of user \(u\), given by
\(
\pmb{H}_u
=
\big(\pmb{F}_N\otimes\pmb{G}_{\text{rx}}\big)
\left(
\sum_{q=1}^{Q}
h_{u,q}
\pmb{\Pi}^{l_{u,q}}
\pmb{\Delta}^{k_{u,q}}
\right)
\big(\pmb{F}_N^H\otimes\pmb{G}_{\text{tx}}\big)
\).
Here, \(\pmb{F}_N\) is the \(N\times N\) normalized discrete Fourier transform (DFT) matrix, \(\pmb{G}_{\text{tx}}=\pmb{G}_{\text{rx}}=\pmb{I}_M\) under rectangular transmit and receive pulse shapes, and \(\pmb{\Pi}\) and \(\pmb{\Delta}\) denote the delay-domain cyclic shift matrix and Doppler-domain phase-shift matrix, respectively~\cite{OTFS_superimposed1}. In addition, \(\pmb{x}_{\text{c}}=\operatorname{vec}(\pmb{X}_{\text{c}})\) and \(\pmb{x}_{\text{p},j}=\operatorname{vec}(\pmb{X}_{\text{p},j})\) denote the vectorized common message and private message of user \(j\), respectively, and \(\pmb{n}_u\sim\mathcal{CN}(\pmb{0}_{MN},\sigma_n^2\pmb{I}_{MN})\) is the noise vector with variance \(\sigma_n^2\). We use column-wise vectorization, where the \((l,k)\)-th delay--Doppler index is mapped to the vector index \(\nu(l,k)=l+1+kM{\in\{1,\ldots, MN\}}\), so that \(\pmb{x}_{\text{c}}(\nu(l,k))=x_{\text{c}}[l,k]\). Accordingly, the vector-index sets corresponding to the pilot, guard, and data positions are defined as \(\mathcal{A}_{\chi}=\{\nu(l,k):(l,k)\in\mathcal{A}_{\chi}^{\text{DD}}\}\), where \(\chi\in\{\text{pilot},\text{guard},\text{data}\}\).

\section{Proposed Framework}
In this section, we describe the channel estimation and data detection procedures. The structural asymmetry of the common message--comprising pilot, guard, and common data--along with the interference from private messages during channel estimation, necessitates a tailored processing strategy distinct from conventional OTFS systems. We first present the channel estimation procedure based on the LMMSE estimator \cite{estimation_theory}, and then describe the decoding of the common and private messages using a message-passing algorithm \cite{OTFS_channel_est1}.
\subsection{Channel Estimation}
Upon reception, each user estimates the channel using the pilot symbol embedded in the common message. For channel estimation, we define the observation index set in the delay--Doppler domain as
\(\mathcal{A}_{\text{obs}}^{\text{DD}}
=
\{(l,k): l_{\text{r}}\le l\le l_{\text{r}}+l_{\text{max}},
k_{\text{r}}-k_{\text{max}}\le k\le k_{\text{r}}+k_{\text{max}}\}\),
whose cardinality is \(N_{\text{obs}}=(2k_{\text{max}}+1)(l_{\text{max}}+1)\).
Using the vector-index mapping \(\nu(l,k)\), the corresponding vector-index set is defined as
\(\mathcal{A}_{\text{obs}}
=
\{\nu(l,k):(l,k)\in\mathcal{A}_{\text{obs}}^{\text{DD}}\}\).
The observation vector \(\bar{\pmb{y}}_u\in\mathbb{C}^{N_{\text{obs}}}\) is then formed by collecting the entries of \(\pmb{y}_u\) indexed by \(\mathcal{A}_{\text{obs}}\), and can be expressed as 
\begin{align}
\label{estimation_equation}
    \bar{\pmb{y}}_u 
    &= \bar{\pmb{h}}_u x_{\text{c,r}}
      + \bar{\pmb{H}}_u \pmb{x}_{\text{p}}
      + \bar{\pmb{n}}_u
      =
      \bar{\pmb{h}}_u x_{\text{c,r}} + \tilde{\pmb{n}}_u,
\end{align}
where \(\bar{\pmb{h}}_u \in \mathbb{C}^{N_{\text{obs}}}\) is the pilot-induced channel vector over the selected observation indices. Its \(Q\) entries corresponding to the actual delay--Doppler taps are determined by the elements of \(\pmb{h}_u\) up to known deterministic phase shifts, while all the remaining entries are zero. Moreover, \(\bar{\pmb{H}}_u \in \mathbb{C}^{N_{\text{obs}} \times MN}\) is the channel matrix associated with the aggregate private-message vector \(\pmb{x}_{\text{p}}=\sum_{j=1}^{U}\pmb{x}_{\text{p},j}\), \(\bar{\pmb{n}}_u\sim \mathcal{CN}(\pmb{0}_{N_{\text{obs}}}, \sigma_n^2 \pmb{I}_{N_{\text{obs}}})\) is the noise vector, and \(\tilde{\pmb{n}}_u = \bar{\pmb{H}}_u \pmb{x}_{\text{p}} + \bar{\pmb{n}}_u\) denotes the effective noise vector.
 Before presenting the channel estimation procedure, we first characterize the distribution of the effective noise in the following lemma.
\begin{lemma}
\label{lemma1}
The effective noise vector \(\tilde{\pmb{n}}_u\) can be approximated as
\begin{align}
\tilde{\pmb{n}}_u
\sim
\mathcal{CN}\!\left(
\pmb{0}_{N_{\text{obs}}},
(P_{\text{p}}\sigma_u^2+\sigma_n^2)\pmb{I}_{N_{\text{obs}}}
\right).
\end{align}
\end{lemma}

\begin{proof}
From \(\tilde{\pmb{n}}_u=\bar{\pmb{H}}_u\pmb{x}_{\text{p}}+\bar{\pmb{n}}_u\),
the aggregate private-message interference is zero-mean and can be
approximated as Gaussian by the central limit theorem~\cite{OTFS-central}.
Moreover, since
\(\mathbb{E}[\pmb{x}_{\text{p}}\pmb{x}_{\text{p}}^{H}]
=P_{\text{p}}\pmb{I}_{MN}\) and
\(\mathbb{E}[\bar{\pmb{H}}_u\bar{\pmb{H}}_u^{H}]
=\sigma_u^2\pmb{I}_{N_{\text{obs}}}\), we have
\begin{align}
\mathbb{E}\!\left[
\tilde{\pmb{n}}_u\tilde{\pmb{n}}_u^{H}
\right]
=
(P_{\text{p}}\sigma_u^2+\sigma_n^2)
\pmb{I}_{N_{\text{obs}}},
\end{align}
which completes the proof.
\end{proof}

Based on Lemma~\ref{lemma1}, the effective noise is characterized as isotropic and can be modeled as white Gaussian noise. Accordingly, the threshold-based channel tap estimation method in \cite{OTFS_channel_est1} can be applied to identify the delay--Doppler support from $\bar{\pmb{y}}_u$ based on \eqref{estimation_equation}. This yields the corresponding delay and Doppler indices \cite{OTFS_channel_est1, OTFS_channel_est2}. Since this procedure is well established in prior work, we omit further details for brevity.
{Motivated by the near-perfect delay--Doppler support detection demonstrated at moderate signal-to-noise ratio (SNR) in~\cite{OTFS_superimposed1}, we assume that the support of the $Q$ channel taps is perfectly estimated.} Under this assumption, we extract the \(Q\)-dimensional subvector {\(\tilde{\pmb{y}}_u\)}  of
\(\bar{\pmb{y}}_u\) corresponding to the estimated delay--Doppler support and
index it consistently with \(\pmb{h}_u\). Since the OTFS path variances follow
a delay-dependent power profile~\cite{OTFS_channel_est1,OTFS_superimposed1},
\(\tilde{\pmb{y}}_u(q)\) is associated with \(h_{u,q}\) and its variance
\(\sigma_{u,q}^2\) according to the corresponding delay index. The channel
gains are then estimated using the LMMSE estimator~\cite{estimation_theory}. The known path-dependent phases are compensated when forming
$\tilde{\pmb {y}}_u$, and the pilot is chosen as
$x_{\text{c,r}}=\sqrt{P_{\text{c,r}}}$.
\begin{theorem}
\label{theorem 1}
{Based on \(\tilde{\pmb{y}}_u\), the LMMSE estimate of the \(q\)-th element of \(\pmb{h}_u\) is given by}
\begin{align}
\hat{\pmb{h}}_u(q)
=
\frac{\sqrt{P_{\text{c,r}}}\sigma_{u,q}^2}
{P_{\text{c,r}}\sigma_{u,q}^2+\tilde{\sigma}_{n,u}^2}
\tilde{\pmb{y}}_u(q),
\label{eq:LMMSE_est}
\end{align}
where \(\tilde{\sigma}_{n,u}^2=P_{\text{p}}\sigma_u^2+\sigma_n^2\).  
Furthermore, the channel estimation error covariance matrix is given by the following diagonal matrix:
\begin{align}
\label{channel_error_matrix}
    \pmb{\Sigma}_{u}=\mathbb{E}[\pmb{h}_u^{\text{err}} \pmb{h}_u^{\text{err}\,H}]
    = \text{diag}((\sigma_{u,1}^{\text{err}})^2, \dots, (\sigma_{u,Q}^{\text{err}})^2),
\end{align}
where $\pmb{h}_u^{\text{err}}={\pmb{h}}_u-\hat{\pmb{h}}_u$ is the channel error vector. Here, $(\sigma_{u,q}^{\text{err}})^2$ is given as
\begin{align}
    (\sigma_{u,q}^{\text{err}})^2
    = \frac{\tilde{\sigma}_{n,u}^2 \, \sigma_{u,q}^2}
           {P_{\text{c,r}} \sigma_{u,q}^2 + \tilde{\sigma}_{n,u}^2}.
    \label{eq:LMMSE_MSE}
\end{align}
\end{theorem}
\begin{proof}
The input--output relation can be written as $\tilde{\pmb{y}}_u={\pmb{h}}_u x_{\text{c,r}} + \dot{\pmb{n}}_u,$ 
where $\dot{\pmb{n}}_u\in \mathbb{C}^Q$ is the associated subvector of $\tilde{\pmb{n}}_u$.
By Lemma 1, the effective noise vector $\tilde{\pmb{n}}_u$ is characterized as isotropic. Therefore, the subvector $\dot{\pmb{n}}_u$ extracted from $\tilde{\pmb{n}}_u$ can be modeled as $\dot{\pmb{n}}_u \sim \mathcal{\mathcal{CN}}\!\left(\pmb{0}_Q,\, \left(P_{\text{p}}\sigma_u^2 + \sigma_n^2\right)\pmb{I}_Q\right).$ Therefore, the LMMSE estimator for the channel vector $ {\pmb{h}}_u $ given the observation $\tilde{\pmb{y}}_u $ is
\begin{align}
    \hat{\pmb{h}}_u
    = \mathbb{E}\!\left[ {\pmb{h}}_u \tilde{\pmb{y}}_u^{H} \right]
    \left( \mathbb{E}\!\left[ \tilde{\pmb{y}}_u \tilde{\pmb{y}}_u^{H} \right] \right)^{-1} \tilde{\pmb{y}}_u.
\end{align}
The channel taps are uncorrelated, and the effective noise
has covariance $\widetilde{\sigma}_{n,u}^{2}\mathbf I_Q$
and is uncorrelated with the channel vector. Therefore,
the LMMSE estimator decouples across taps. Hence, for each tap indexed by $ q $,
\begin{align}
    \hat{\pmb{h}}_u(q)
=
\frac{\sqrt{P_{\text{c,r}}}\sigma_{u,q}^2}
{P_{\text{c,r}}\sigma_{u,q}^2+\tilde{\sigma}_{n,u}^2}
\tilde{\pmb{y}}_u(q).
\end{align}

Since each tap is estimated independently, the corresponding estimation error covariance matrix is given by \eqref{channel_error_matrix}.
\end{proof}
{Building on the closed-form MSE expression derived above, we account for  the residual interference resulting from imperfect SIC due to channel estimation errors.} 
This constitutes a key distinction of our work from many existing studies, which ignore channel uncertainty and the resulting SIC error. 
\subsection{Message-Passing-Based Data Detection}
Based on \eqref{receive signal}, equalization can be performed using the inverse of the channel matrix $\pmb{H}_u$ to recover the transmitted symbols. However, since the effective OTFS channel matrix has a size of $MN \times MN$, this approach incurs a computational complexity of $\mathcal{O}((MN)^3)$, rendering it impractical for large-scale systems. To address this issue, we adopt a message-passing-based decoding approach \cite{OTFS_superimposed1}, following the low-complexity detection framework commonly used in conventional OTFS systems. {The message-passing algorithm constructs a factor graph consisting of
variable nodes for transmitted data symbols and observation nodes for
the entries of the received vector.}
Data detection is then performed through iterative message exchanges
between connected nodes until convergence.
The detailed procedure is described in the following.

First, each user decodes the common data before decoding its private message. The received vector at user \(u\) from the perspective of common-data decoding can be expressed as follows: 
\begin{align}
\label{commonmessage}
\pmb{y}_{u}
&=
\widehat{\pmb{H}}_u \pmb{x}_{\text{c}}
+
\underbrace{\pmb{H}_u^{\text{err}}\pmb{x}_{\text{c}}}_{\text{estimation error}}
+
\underbrace{\pmb{H}_u \pmb{x}_{\text{p}}}_{\text{private interference}}
+
\pmb{n}_u,
\end{align}
where $\hat{\pmb{H}}_u
    = ( \pmb{F}_N \otimes \pmb{G}_{\text{rx}} )
      \left( \sum_{q=1}^{Q} \hat{h}_{u,q}\,
             \pmb{\Pi}^{\,\hat{l}_{u,q}} \pmb{\Delta}^{\,\hat{k}_{u,q}} \right)
      ( \pmb{F}_N^H \otimes \pmb{G}_{\text{tx}} )$
is the estimated delay-Doppler channel matrix of user $u$. Here, $\hat{l}_{u,q}$ and $\hat{k}_{u,q}$ denote the estimated delay and Doppler indices of the $q$-th tap for user $u$, respectively. Moreover,
$\pmb{H}_u^{\text{err}}
    = ( \pmb{F}_N \otimes \pmb{G}_{\text{rx}} )
      \left( \sum_{q=1}^{Q} (h_{u,q} - \hat{h}_{u,q})\,
             \pmb{\Pi}^{\,\hat{l}_{u,q}} \pmb{\Delta}^{\,\hat{k}_{u,q}} \right)
      ( \pmb{F}_N^H \otimes \pmb{G}_{\text{tx}} )$
is the delay-Doppler channel error matrix. 
We define
\(\tilde{\pmb{n}}_{\text{c},u}
=
\pmb{H}_u^{\text{err}}\pmb{x}_{\text{c}}
+
\pmb{H}_u\pmb{x}_{\text{p}}
+
\pmb{n}_u\)
as the effective noise vector in common-data decoding. Accordingly,
\(\pmb{y}_u(a)\) is determined by \(Q\) elements of
\(\pmb{x}_{\text{c}}\) through the sparse OTFS channel and the effective
noise term \(\tilde{\pmb{n}}_{\text{c},u}(a)\). Based on this vectorized
representation, we construct a message-passing factor graph by treating
\(\pmb{y}_u(a)\) as the \(a\)-th observation node and
\(\pmb{x}_{\text{c}}(b)\) as the \(b\)-th variable node. The corresponding input-output relation can be written as
\begin{align}
    &\pmb{y}_u(a) = \sum_{b\in \mathcal{I}_u(a)}\hat{\pmb{H}}_u(a,b) \pmb{x}_{\text{c}}(b)+\tilde{\pmb{n}}_{\text{c},u}(a) \notag \\ &= \hat{\pmb{H}}_u(a,b) \pmb{x}_{\text{c}}(b)+\sum_{d\in \mathcal{I}_u(a), d\neq b }\hat{\pmb{H}}_u(a,d) \pmb{x}_{\text{c}}(d)+\tilde{\pmb{n}}_{\text{c},u}(a),
    \notag
\end{align}
where {$\mathcal{I}_u(a)$} denotes the set of column indices corresponding to the nonzero entries in the $a$-th row of $\hat{\pmb{H}}_u$. Similarly, let $\mathcal{J}_u(b)$ denote the set of row indices corresponding to the nonzero entries in the $b$-th column.

In the message-passing algorithm, each observation node sends the mean and variance of the residual interference to its connected variable nodes at every iteration. {For notational simplicity, the user index \(u\) is omitted in the following message-passing variables.} Let \(\mu_{a,b}^{(i)}\) and \((\sigma_{a,b}^{(i)})^2\) denote the mean and variance of the residual interference at observation node \(a\), excluding variable node \(b\), in the \(i\)-th iteration. Following the Gaussian message-passing approximation in~\cite{OTFS_superimposed1}, they are given by
\begin{align}
    \mu_{a,b}^{(i)}
    = \sum_{\substack{d \in \mathcal{I}(a) \\ d \neq b}}
      & \sum_{j=0}^{S}
        p_{{d},a}^{(i-1)}(\alpha_j)\,\alpha_{j}\,\hat{\pmb{H}}_u(a,d)
       + \mathbb{E}[\tilde{\pmb{n}}_{\text{c},u}(a)], \label{common_mu}\\[4pt]
(\sigma_{a,b}^{(i)})^{2} 
    =& \sum_{\substack{d \in \mathcal{I}(a) \\ d \neq b}}
       \sum_{j=0}^{S}
        p_{d,a}^{(i-1)}(\alpha_j)\,|\alpha_{j}|^2\,\big|\hat{\pmb{H}}_u(a,d)\big|^2 \notag \\
       & - |\mu_{a,b}^{(i)}|^2
       + \mathbb{V}[\tilde{\pmb{n}}_{\text{c},u}(a)] \label{common_var},
\end{align}
{where \(p_{d,a}^{(i-1)}(\alpha_j)\) denotes the probability mass function (PMF) passed from variable
node \(d\) to observation node \(a\) for symbol \(\alpha_j\) at the
\((i-1)\)-th iteration.
To account for the zero-valued guard symbols within the
message-passing framework, we augment the data constellation
\(\mathcal{S}\) with the zero symbol, i.e.,
\(\alpha_j\in\{0\}\cup\mathcal{S}\).
Owing to the guard region, the pilot symbol \(x_{\text{c,r}}\) does not
interfere with the common-data symbols during common-data detection.
Thus, for notational convenience, the pilot position is also represented
by the zero symbol.}
Accordingly, the PMFs are initialized as
\(p_{d,a}^{(0)}(\alpha_j)=1/S\) for
\(d\in\mathcal{A}_{\text{data}}\) and \(\alpha_j\in\mathcal{S}\), whereas
\(p_{d,a}^{(0)}(0)=1\) for
\(d\in\mathcal{A}_{\text{pilot}}\cup\mathcal{A}_{\text{guard}}\), with all remaining probabilities set to zero. Starting
from these initial PMFs, the mean and variance \(\mu_{a,b}^{(i)}\) and
\((\sigma_{a,b}^{(i)})^2\) are computed using \(p_{d,a}^{(i-1)}(\alpha_j)\)
together with the mean and variance of
\(\tilde{\pmb{n}}_{\text{c},u}(a)\), which are derived next. These
mean--variance messages are then used to update \(p_{b,a}^{(i)}(\alpha_j)\),
as described later in this section.



We now derive the mean and variance of $\tilde{\pmb{n}}_{\text{c},u}(a)$. 
From the adopted channel model, all channel coefficients are zero-mean. In addition, the common data and private message symbols are assumed to be zero-mean and $ \pmb{n}_u $ is additive white Gaussian noise with zero mean.
Therefore,
$ \mathbb{E}[\tilde{\pmb{n}}_{\text{c},u}(a)] = {0}, \forall a.$
Hence, the variance of $\tilde{\pmb{n}}_{\text{c},u}(a)$ coincides with its second moment:
\begin{align}
    \mathbb{V}[\tilde{\pmb{n}}_{\text{c},u}(a)]
    =\sum_{b\in\mathcal{I}(a)} {\pmb{\Sigma}}_u^{\text{err}}(a,b) {P}_{\text{c}}(b)+ P_{\text{p}} \textstyle\sum_{q} \sigma_{u,q}^{2} + \sigma_n^{2}, {\forall a},
    \notag
\end{align}
where $\pmb{\Sigma}_u^{\text{err}}(a,b)
    = \mathbb{E}[|H^{\text{err}}_u(a,b)|^2]$
denotes the element-wise channel-error variance and $P_{\text{c}}(b)$ denotes the position-dependent common-symbol power given as
\begin{align}
    {P}_{\text{c}}(b)
    =
    \begin{cases}
        P_{\text{c,r}}, & b \in \mathcal{A}_{\text{pilot}}, \\[2pt]
        0,       & b \in \mathcal{A}_{\text{guard}}, \\[2pt]
        P_{\text{c,d}}, & b \in \mathcal{A}_{\text{data}}.
    \end{cases}
\end{align}
At each variable node in the \(i\)-th iteration, the received mean and
variance, {$\mu_{a,b}^{(i)}$ and $  (\sigma_{a,b}^{(i)})^{2}$}, are used to compute the PMF \(p_{b,a}^{(i)}(\alpha_j)\). 
Specifically, the PMF is given by
\begin{align}
\label{PMF_calculate}
p_{b,a}^{(i)}(\alpha_j)
=
\begin{cases}
\displaystyle
\frac{
\prod_{c \in \mathcal{J}_{u,b}^{-a}}\beta_{c,b,j}
}{
\sum_s
\prod_{c \in \mathcal{J}_{u,b}^{-a}}\beta_{c,b,s}
},
& \!b \in \! \mathcal{A}_{\text{data}}, \\[12pt]
1,
& \!b \! \in\! \mathcal{A}_{\text{guard}}\cup\mathcal{A}_{\text{pilot}},\ \alpha_j = 0, \\[2pt]
0,
& \!b \!\in \!\mathcal{A}_{\text{guard}}\cup\mathcal{A}_{\text{pilot}},\ \alpha_j \neq 0,
\end{cases}
\end{align}
where
\(\beta_{c,b,s}
= \exp\!\left(
-\frac{|\pmb{y}_u(c)-\mu_{c,b}^{(i)}
-\widehat{\pmb{H}}_u(c,b)\alpha_s|^2}
{(\sigma_{c,b}^{(i)})^{2}}
\right)\) and \(\mathcal{J}_{u,b}^{-a}\triangleq \mathcal{J}_u(b)\setminus\{a\}\).
To improve stability and convergence, we apply a damping scheme~\cite{damping},
yielding
\begin{align}
\label{PMF_damping}
p_{b,a}^{(i)}(\alpha_j)
\leftarrow
\Delta\, p_{b,a}^{(i-1)}(\alpha_j)
+
(1-\Delta)\, p_{b,a}^{(i)}(\alpha_j),
\end{align}
where \(\Delta \in [0,1)\) is the damping factor. The resulting edge PMF
\(p_{b,a}^{(i)}(\alpha_j)\) is then passed to the connected observation nodes. {For convergence checking and symbol decision, we define
\(p_b^{(i)}(\alpha_j)\) as the marginal PMF of
\(\alpha_j\in\mathcal{S}\) at variable node \(b\) in the \(i\)-th
iteration, obtained by combining all incoming messages from its
connected observation nodes as
\begin{align}
\label{marginal_PMF}
p_b^{(i)}(\alpha_j)
=
\frac{
\prod_{c\in\mathcal{J}(b)} \beta_{c,b,j}
}{
\sum_{s}
\prod_{c\in\mathcal{J}(b)} \beta_{c,b,s}
},
\quad \alpha_j\in\mathcal{S}.
\end{align}}
These mean–variance messages and PMFs are iteratively exchanged between the
observation and variable nodes. We define the convergence criterion as
\begin{align}
\label{convergence_criterion}
    \zeta^{(i)}
    = \frac{1}{N_{\text{c}}} \sum_{b\in \mathcal{A}_{\text{data}}}
      \mathbb{I}\!\left(
        \max_{\alpha_j \in \mathcal{S}} p_b^{(i)}(\alpha_j) \ge 1 - \epsilon
      \right),
\end{align}
where $N_{\text{c}}=MN-N_{\text{g}}-1$ is the number of common data symbols, and $\mathbb{I}(\cdot)$ denotes the indicator function. The algorithm terminates when $\zeta^{(i)} \ge \zeta_{\text{th}}$ or $i = i_{\text{max}}$, where $\zeta_{\text{th}} \in (0,1]$ is a predefined threshold and $i_{\text{max}}$ denotes the maximum number of iterations. The convergence criterion is evaluated only over the common-data positions.  Once convergence is achieved, the detected common-data symbol at the \(b\)-th variable node, \(b\in\mathcal{A}_{\text{data}}\), is chosen as
\begin{align}
\label{hard_decision}
    \hat{\pmb{x}}_{\text{c}}(b)
    =
    \arg\max_{\alpha_j\in\mathcal{S}} p_b^{(i)}(\alpha_j).
\end{align}
The message-passing algorithm for common-data detection, referred to as \textbf{MP-C}, is summarized in Algorithm~\ref{alg:message passing}.

\begin{algorithm}[t]
\caption{\\\textbf{~~Message Passing for Common-Data Detection (MP-C)}}
\label{alg:message passing}
\begin{algorithmic}[1]
\STATE \textbf{Input:} $\mathcal{S}$, $\pmb{y}_u$, $\hat{\pmb{H}}_u$, $\pmb{K}_{h_u h_u}$, $\pmb{\Sigma}_{u}$, $P_{\text{c,r}}$, $P_{\text{c,d}}$, $P_{\text{p},u}$,  $P_{\text{p}}$, $\sigma_n^2$, $\Delta$, $\epsilon$, $\zeta_{\text{th}}$ and $i_{\text{max}}$.
\STATE \textbf{Initialization:} Set \(i=1\). Initialize
\(p_{b,a}^{(0)}(\alpha_j)=1/|\mathcal{S}|\) for
\(b\in\mathcal{A}_{\text{data}}\), \(\alpha_j\in\mathcal{S}\), and
\(p_{b,a}^{(0)}(0)=1\) for
\(b\in\mathcal{A}_{\text{pilot}}\cup\mathcal{A}_{\text{guard}}\);
set all unspecified PMF entries to zero.
\REPEAT
    \STATE Compute $\mu_{a,b}^{(i)}$ and $(\sigma_{a,b}^{(i)})^2$ using \eqref{common_mu} and \eqref{common_var}.
    \STATE Calculate $\pmb{p}_{b,a}^{(i)}$ using \eqref{PMF_calculate} and update by \eqref{PMF_damping}.
    \STATE $i \gets i+1$
\UNTIL{$\zeta^{(i-1)} \ge \zeta_{\text{th}}$ or $i = i_{\text{max}}$}.
\STATE Compute symbol estimates $\hat{\pmb{x}}_c$ by \eqref{hard_decision}.
\STATE \textbf{Output:} $\hat{\pmb{x}}_c$
\end{algorithmic}
\end{algorithm}

For private message decoding, a similar message-passing procedure is applied. After removing the decoded common message via SIC, the resulting detection model retains the same structure as in the common data case, with a modified effective noise term. Under the assumption of perfect common data decoding, the received vector at user $u$ from the perspective of private message decoding can be written as 
\begin{align}
\label{private_message}
    \pmb{y}_{\text{p},u}&= \pmb{y}_{u}- \hat{\pmb{H}}_u \pmb{x}_{\text{c}}= 
    \hat{\pmb{H}}_u \pmb{x}_{\text{p},u}
    + \tilde{\pmb{n}}_{\text{p},u},
\end{align}
with the effective noise  $\tilde{\pmb{n}}_{\text{p},u}=\pmb{H}_u^{\text{err}}
      \big(\pmb{x}_{\text{c}} + \pmb{x}_{\text{p},u}\big)
    + \pmb{H}_u\big(\pmb{x}_{\text{p}} - \pmb{x}_{\text{p},u}\big)
    + \pmb{n}_u$. By formulating the message passing framework in the same manner as in the common data decoding, the input–output relation between the $a$-th observation node $\pmb{y}_{\text{p},u}(a)$ and the $b$-th variable node $\pmb{x}_{\text{p},u}(b)$ can be written as
\begin{align}
    &\pmb{y}_{\text{p},u}(a) = \sum_{b\in \mathcal{I}(a)}\hat{\pmb{H}}_u(a,b) \pmb{x}_{\text{p},u}(b)+\tilde{\pmb{n}}_{\text{p},u}(a) \notag \\ &= \hat{\pmb{H}}_u(a,b) \pmb{x}_{\text{p},u}(b)+\sum_{d\in \mathcal{I}(a), d\neq b }\hat{\pmb{H}}_u(a,d) \pmb{x}_{\text{p},u}(d)+\tilde{\pmb{n}}_{\text{p},u}(a).
    \notag
\end{align}
{Similar to the common-data case, the mean and variance of the effective
noise for private-message decoding are given by}
\begin{align}
    &\mathbb{E}[\tilde{\pmb{n}}_{\text{p},u}(a)] = {0},\qquad   \\
        &\mathbb{V}[\tilde{\pmb{n}}_{\text{p},u}(a)] \notag \\  &=\!\sum_{b\in\mathcal{I}(a)} \!{\Sigma}_u^{\text{err}}(a,\!b) \big(P_{\text{c}}(b)\! +\! P_{\text{p},u})
        + \sigma_{u}^{2}\big(P_{\text{p}} \!-\! P_{\text{p},u}) \!+\! \sigma_n^{2}.
\end{align}
All the remaining steps are identical to those of common-data decoding.
Specifically, in \eqref{common_mu} and \eqref{common_var},
\(\tilde{\pmb{n}}_{\text{c},u}(a)\) is replaced by
\(\tilde{\pmb{n}}_{\text{p},u}(a)\). The resulting mean--variance messages and
PMFs are iteratively exchanged until convergence, where the criterion in
\eqref{convergence_criterion} is evaluated by replacing
\(\mathcal{A}_{\text{data}}\) and \(N_{\text{c}}\) with
\(\{1,\ldots,MN\}\) and \(MN\), respectively. After convergence, the transmitted
symbols are estimated according to \eqref{hard_decision}, with
\(\hat{\pmb{x}}_{\text{c}}(b)\) replaced by
\(\hat{\pmb{x}}_{\text{p},u}(b)\). {Accordingly, the private-message detection algorithm, referred to as
\textbf{MP-P}, follows Algorithm~\ref{alg:message passing} with minor
modifications and is omitted for brevity.} 
\subsection{Overall Algorithm}
The overall procedure of the proposed OTFS--RSMA receiver is summarized as follows. First, each user estimates its delay--Doppler channel using the pilot symbol embedded in the common message and constructs the estimated channel matrix. Based on the estimated channel, the common message is decoded by the proposed \textbf{MP-C} detector and then canceled via SIC. Finally, each user decodes its own private message using the proposed \textbf{MP-P} detector.

{We analyze the computational complexity of the proposed channel estimation and decoding procedures. For each user, the \(Q\) delay--Doppler channel taps are estimated independently using scalar LMMSE operations, resulting in \(\mathcal{O}(Q)\) complexity, while the message-passing graph construction requires \(\mathcal{O}(MNQ)\) operations because each delay--Doppler node is connected to \(Q\) path-induced neighboring nodes. The dominant complexity comes from the message-passing detector, whose factor graph contains \(\mathcal{O}(MNQ)\) edges. Since each edge update in \eqref{common_mu} and \eqref{common_var} involves \(Q\) neighboring nodes and \(S\) constellation points, one iteration requires \(\mathcal{O}(MNQ^2S)\) operations. The final symbol decision requires only \(\mathcal{O}(MNS)\) operations and is negligible, and the SIC operation based on the sparse multiplication \(\hat{\pmb{H}}_u\hat{\pmb{x}}_{\text{c}}\) requires \(\mathcal{O}(MNQ)\) operations. Since both common-data and private-message decoding use the same message-passing procedure, the overall per-user complexity is \(\mathcal{O}(Q+MNQ+2IMNQ^2S+MNQ)=\mathcal{O}(IMNQ^2S)\), where \(I\) denotes the number of message-passing iterations.}
\begin{remark}
{Although the proposed channel-estimation and message-passing detection algorithms are developed for the GS configuration, they can be readily
extended to the other three configurations by modifying the corresponding effective-noise and symbol-region models.
For channel estimation, changing the private-message structure from superimposed to guard-based, i.e., from GS to GG, removes the $P_p\sigma_u^2$ interference term from
$\widetilde{\sigma}_{n,u}^2$, whereas changing the common-message structure
from guard-based to superimposed, i.e., from GS to SS, introduces the $P_{c,d}\sigma_u^2$ interference term.
For SG, both modifications apply simultaneously.
The message-passing detectors can be extended analogously by adjusting the configuration-dependent symbol regions and effective-noise terms, without changing their basic algorithmic structure.}
\end{remark}

\section{{Resource Allocation}}
The performance of the proposed OTFS--RSMA system critically depends on
the allocation of transmit power among the pilot, common data, and private
messages, as well as the allocation of the common rate among users.
In particular, these allocations should jointly balance channel estimation
accuracy and interference management while satisfying the users' rate
requirements under a limited transmit power budget.
Accordingly, in this section, we investigate the joint power and rate
allocation for sum-rate maximization subject to a total transmit power
constraint and per-user minimum-rate constraints.
We first derive the achievable common and private message rates for each user.
Since direct optimization of these rate expressions is intractable, we
introduce tractable surrogate rate functions and formulate the corresponding
joint allocation problem based on them.
We then develop an SCA-based algorithm \cite{boyd2004convex} to solve the
resulting problem. To further reduce the computational complexity, we develop
a low-complexity algorithm based on a conservative rate approximation and
analytical characterization.
\subsection{Problem Formulation}
First, we derive the rate expressions for the common and private messages. From \eqref{commonmessage} and \eqref{private_message}, the received vectors are given by $\pmb{y}_{u} = \hat{\pmb{H}}_u \pmb{x}_{\text{c}} + \tilde{\pmb{n}}_{\text{c},u}$ and $\pmb{y}_{\text{p},u} = \hat{\pmb{H}}_u \pmb{x}_{\text{p},u} + \tilde{\pmb{n}}_{\text{p},u}$. For tractable rate analysis, we adopt a Gaussian signaling-based rate model, in which the transmitted data symbols are modeled as Gaussian and the effective residual interference-plus-noise terms are approximated as Gaussian. Then, the rates of the common and private messages are given by \cite{tse2005fundamentals}
\begin{align}
R_{\text{c},u} &= \frac{1}{MN}\mathbb{E}_{\pmb{H}} \left[\log_2 \det\! \left(
\pmb{I}
+\hat{\pmb{H}}_u \pmb{K}_{{\text{c,d}}} \hat{\pmb{H}}_u^H
\pmb{K}_{\text{c},u}^{-1}
\right)\right],\label{common_rate} \\
R_{\text{p},u} &= \frac{1}{MN}\mathbb{E}_{\pmb{H}} \left[
\log_2\det \left(
\pmb{I}
+ P_{\text{p},u} \, \hat{\pmb{H}}_u \hat{\pmb{H}}_u^{H}
\pmb{K}_{\text{p},u}^{-1}
\right)
\right],\label{private_rate}
\end{align}
where
\(\pmb{K}_{\text{c,d}} = P_{\text{c,d}}\pmb{I}_{MN}
-\sum_{i\in\mathcal{A}_{\text{guard}}}
P_{\text{c,d}}\pmb{e}_{i}\pmb{e}_{i}^{H}
-P_{\text{c,d}}\pmb{e}_{r}\pmb{e}_{r}^{H}\)
is the covariance matrix of the common-data vector. Moreover,
\(\pmb{K}_{\text{c},u}\) and \(\pmb{K}_{\text{p},u}\), given in
\eqref{common effective noise} and \eqref{private effective noise}, are the
covariance matrices of \(\tilde{\pmb{n}}_{\text{c},u}\) and
\(\tilde{\pmb{n}}_{\text{p},u}\), respectively. Here,
\(\pmb{e}_i\in\mathbb{R}^{MN}\) denotes the standard basis vector whose
\(i\)-th entry is one, and \(r=\nu(l_{\text{r}},k_{\text{r}})\) denotes the
pilot vector index. The expectation is taken over the random channel
realization and the corresponding estimated channel matrix
\(\hat{\pmb{H}}_u\).                                                                                                                              
\begin{figure*}[t]
\begin{align}
\label{common effective noise}
\mathbf K_{\text c,u}
&=
\operatorname{diag}
\left(
\left\{
\sum_{b\in\mathcal I(a)}
\Sigma_u^{\text{err}}(a,b)P_c(b)
+
P_{\text p}\sigma_u^2+\sigma_n^2
\right\}_{a=1}^{MN}
\right),\\
\label{private effective noise}
\mathbf K_{\text p,u}
&=
\operatorname{diag}
\left(
\left\{
\sum_{b\in\mathcal I(a)}
\Sigma_u^{\text{err}}(a,b)P_c(b)
+
P_{\text p,u}\operatorname{Tr}(\boldsymbol{\Sigma}_u)
+
(P_{\text p}-P_{\text p,u})\sigma_u^2
+
\sigma_n^2
\right\}_{a=1}^{MN}
\right).
\end{align}
\hrulefill
\end{figure*}
However, direct optimization of the rate expressions in \eqref{common_rate} and \eqref{private_rate} is intractable due to the expectation over the random channel $\mathbb{E}_{\pmb{H}}$. To address this issue, we apply Jensen's inequality~\cite{boyd2004convex} to construct tractable surrogate rate functions, given by
\begin{align}
R_{\text{c},u}^{\text{s}} &=  \frac{1}{MN}\log_2 \det\! \left(
\pmb{I}
+\mathbb{E}_{\pmb{H}}[\hat{\pmb{H}}_u \pmb{K}_{\text{c,d}} \hat{\pmb{H}}_u^H]
\pmb{K}_{\text{c},u}^{-1}
\right),\label{common_rate_surrogate} \\
R_{\text{p},u}^{\text{s}} &= \frac{1}{MN}
\log_2\det \left(
\pmb{I}
+ P_{\text{p},u} \, \mathbb{E}_{\pmb{H}}[\hat{\pmb{H}}_u \hat{\pmb{H}}_u^{H}]
\pmb{K}_{\text{p},u}^{-1}
\right). \label{private_rate_surrogate}
\end{align}
Since
$\mathbb{E}_{\pmb{H}}[\hat{\pmb{H}}_u \hat{\pmb{H}}_u^H]
=
(\sigma_u^2-\operatorname{Tr}(\pmb{\Sigma}_u))\pmb{I}$,
$\pmb{K}_{\text{c,d}}$ is diagonal, and the OTFS path matrices are monomial with uncorrelated path gains, it follows that
$\mathbb{E}_{\pmb{H}}[\hat{\pmb{H}}_u \pmb{K}_{\text{c,d}} \hat{\pmb{H}}_u^H]$
is also diagonal. Consequently, the Jensen-based surrogate common rate is simplified as
\begin{align}
\label{common surrogate rate}
R_{\text{c},u}^{\text{s}}
=
\frac{1}{MN}
\sum_{a=1}^{MN}
\log_2
\left(
1+
\frac{
\rho_{u,a} 
}{
\displaystyle
\eta_{u,a}
}
\right),
\end{align}
where $\rho_{u,a}=P_{\text{c,d}}
\displaystyle\sum_{d\in\mathcal I_u(a)\cap\mathcal A_{\text{data}}}
\mathbb E_{\mathbf H}
[
|
\widehat{\mathbf H}_u(a,d)
|^2
]$ and $\eta_{u,a}=
\sum_{b\in\mathcal{I}_u(a)}
\Sigma_{u}^{\text{err}}(a,b) P_{\text{c}}(b)
+ \sigma_u^2 P_{\text{p}}
+ \sigma_n^2$.
The surrogate common rate is divided among all users, subject to
\begin{align}
    \sum_{u} C_u \le \min_{u} R_{\text{c},u}^{\text{s}},
    \label{common_mssage_allocaiton}
\end{align}
where $C_u$ is the common-rate portion assigned to user $u$. 
Moreover, the surrogate private rate is expressed as 
\begin{align}
\label{surrogate private rate}
R_{\text{p},u}^{\text{s}}
&=
\frac{1}{MN}
\sum_{a=1}^{MN}
\log_2\!\left(
\frac{\eta_{u,a}}{\eta_{u,a}-\kappa_u P_{\text{p},u}}
\right),
\end{align}
where $\kappa_u=\sigma_u^2-\operatorname{Tr}(\pmb{\Sigma}_u)$. 
\begin{remark}
\label{remark:surrogate} 
The surrogate rates in \eqref{common_rate_surrogate} and
\eqref{private_rate_surrogate} are introduced solely for tractable
power-allocation design; they are not conservative lower bounds on the
ergodic rates. Therefore, the surrogate rate expressions are used only for power-allocation optimization, whereas the numerical results in Section~V report the actual rates evaluated from the original expressions in \eqref{common_rate} and \eqref{private_rate}.
\end{remark}

Based on the surrogate rate expressions, we formulate the joint power and
rate allocation problem to maximize the sum surrogate rate by optimizing
the pilot power, common-data power, private-message powers, and common-rate
allocation, subject to user rate and power constraints, as follows:  
\begin{subequations}
\label{original Problem}
\begin{align}
&\max_{\pmb{z}} && 
\sum_u\big( C_u + R_{\text{p},u}^{\text{s}}(P_{\text{c,r}},P_{\text{c,d}},\pmb{P}_{\text{p}}) \label{objective_main_problem}
\big) \\
& \text{s.t.} && 
 C_u + R_{\text{p},u}^{\text{s}}(P_{\text{c,r}},P_{\text{c,d}},\pmb{P}_{\text{p}}) \ge R^{\text{th}}, \quad \forall u, \label{rate_threshold} \\
&&& P_{\text{c,r}} + MN\sum_u P_{\text{p},u}+N_{\text{c}}P_{\text{c,d}} \le  MNP_{\text{max}}, \label{power_constraint} \\
&&& C_u \ge 0,\ P_{\text{c,d}} \ge 0, P_{\text{c,r}}\ge0, \ P_{\text{p},u} \ge 0, \quad \forall u, \label{non-negtative}\\
&&&\eqref{common_mssage_allocaiton},  
\end{align}
\end{subequations}
where
\(\pmb{z}
=
[P_{\text{c,r}}, P_{\text{c,d}}, \pmb{P}_{\text{p}}^{\mathsf T},
\pmb{C}^{\mathsf T}]^{\mathsf T}\)
is the optimization variable vector. Here,
\(\pmb{P}_{\text{p}}=[P_{\text{p},1},\ldots,P_{\text{p},U}]^{\mathsf T}\)
and \(\pmb{C}=[C_1,\ldots,C_U]^{\mathsf T}\) collect the private-message
powers and common-rate portions of all users, respectively. Constraint \eqref{rate_threshold} ensures the rate requirements, \eqref{power_constraint} limits the average transmit power, and \eqref{non-negtative} ensures nonnegative variables.

Problem \eqref{original Problem} is difficult to solve because the variables are nonlinearly coupled and the resulting optimization is nonconvex. Accordingly, in the following subsections, we first propose a 2D-search-based SCA algorithm to solve the optimization problem. We then develop an additional 2D-search-based low-complexity algorithm to mitigate the high computational burden of the SCA-based approach. The detailed procedure is presented as follows.

\subsection{SCA-Based 2D Grid Search Algorithm}
In this subsection, we solve Problem~\eqref{original Problem} using an
SCA-based 2D grid search. We first reformulate the problem for fixed
pilot and common-data powers, \(P_{\text{c,r}}\) and \(P_{\text{c,d}}\),
based on the following lemma. The resulting problem is solved by SCA,
and the best solution is selected over the 2D grid of
\((P_{\text{c,r}},P_{\text{c,d}})\).
\begin{lemma}
\label{Lemma2}
At the optimum of Problem~\eqref{original Problem}, constraints
\eqref{common_mssage_allocaiton} and \eqref{power_constraint} hold with
equality.
\end{lemma}
\begin{proof}
The result can be proved by contradiction using arguments similar to those in the proof of Proposition 1 in~\cite{RSMA-shared}; hence, the details are omitted for brevity.
\end{proof}
Accordingly, for a fixed feasible candidate pair of pilot and common-data
powers \((P_{\text{c,r}},P_{\text{c,d}})\), Problem~\eqref{original Problem}
can be reformulated as follows:
\begin{subequations}
\label{common message, pilot given problem}
\begin{align}
    \max_{\pmb{P}_{\text{p}}, \pmb{C}} \quad
    & \sum_u\big( C_u + R_{\text{p},u}^{\text{s}}({P}_{\text{p},u}) \big) \\
    \text{s.t.} \quad 
    & \sum_u C_u =  R_{\text{c}}, \label{common rate_powergiven}\\
    & \sum_{u} P_{\text{p},u}
    = \frac{MNP_{\text{max}}-P_{\text{c,r}}-N_{\text{c}}P_{\text{c,d}}}{MN},
    \label{privatepower_rempower_const}\\
    & C_u \geq 0, \quad P_{\text{p},u} \geq 0, \quad \forall u,
    \label{Common,privatepower_positive_const}\\
    & \eqref{rate_threshold}. 
\end{align}
\end{subequations}
where \(R_{\text{c}}=\min_{u}R_{\text{c},u}^{\text{s}}\). For each fixed pair \((P_{\text{c,r}},P_{\text{c,d}})\),
Lemma~\ref{Lemma2} determines the total private-message power
\(P_{\text{p}}\) through \eqref{privatepower_rempower_const}.
Thus, the surrogate private rate of user \(u\) depends only on its
allocated private-message power \(P_{\text{p},u}\), and is denoted by
\(R_{\text{p},u}^{\text{s}}(P_{\text{p},u})\).

The reformulated problem remains nonconvex because the convex function
\(R_{\text{p},u}^{\text{s}}(P_{\text{p},u})\) appears in both the
objective function and the rate constraint in \eqref{rate_threshold}. To handle this nonconvexity, we apply SCA by linearizing $R_{\text{p},u}^{\text{s}}(P_{\text{p},u})$ around a reference point $\tilde{P}_{\text{p},u}$, yielding
\begin{align*}
&\tilde{R}_{\text{p},u}^{s}(P_{\text{p},u})
= \notag \\
&\frac{1}{MN}\!\sum_{a=1}^{MN}
\!\!\left(\frac{\kappa_u (P_{\text{p},u}-\tilde{P}_{\text{p},u})}{(-\kappa_u \tilde{P}_{\text{p},u} \!+\! \eta_{u,a})\!\ln2}\!
+\!
\log_2\!\left(
\frac{\eta_{u,a}}{-\kappa_u \tilde{P}_{\text{p},u}\!+\!\eta_{u,a}}
\right)\right).
\end{align*}
By substituting the first-order approximation $ \tilde{R}_{\text{p},u}^{s}({P}_{\text{p},u}) $ into the original problem, we obtain the following convex optimization problem:
\begin{subequations}
\label{original Problem_Taylor}
\begin{align}
&\max_{\pmb{P}_{\text{p}}, \pmb{C}} && 
\sum_u \left( C_u + \tilde{R}_{\text{p},u}^{s}({P}_{\text{p},u}) \right) \\
& \text{s.t.} && 
C_u + \tilde{R}_{\text{p},u}^{s}({P}_{\text{p},u}) \ge R^{\text{th}}, \quad \forall u, \label{rate_threshold_SCA} \\
&&& 
\eqref{common rate_powergiven}, \eqref{privatepower_rempower_const}, \eqref{Common,privatepower_positive_const}. 
\end{align}
\end{subequations}
Problem~\eqref{original Problem_Taylor} is convex and can be efficiently solved
using standard convex optimization solvers such as CVX~\cite{cvx}.

To optimize the pilot and common-data powers, we perform a 2D grid search
over finite power grids \(\mathcal{P}_{\text{c,r}}\) and
\(\mathcal{P}_{\text{c,d}}\), respectively.
For each feasible pair \((P_{\text{c,r}},P_{\text{c,d}})\), the SCA
procedure is initialized with
\(\tilde P_{\text{p},u}=P_{\text{p}}/U\), \(u\in\mathcal U\), and
Problem~\eqref{original Problem_Taylor} is iteratively solved by updating
\(\tilde P_{\text{p},u}\leftarrow P_{\text{p},u}\) until convergence.
The solution achieving the highest sum rate over all grid points is selected.
The resulting algorithm, referred to as \textbf{SCA-2D}, is summarized in
Algorithm~\ref{alg:2Dsearch}.
\begin{algorithm}[t]
\caption{\textbf{SCA-Based 2D Grid Search Algorithm} (\textbf{SCA-2D}).}
\label{alg:2Dsearch}
\begin{algorithmic}[1]
\STATE \textbf{Input:} $M$, $N$, $N_{\text{c}}$, $N_g$, $P_{\text{max}}$, $R^{\text{th}}$, $\sigma_n^2$, $\epsilon$, $\pmb{K}_{h_u,h_u}$, $\mathcal{P}_{\text{c,r}}$, and $\mathcal{P}_{\text{c,d}}$.
\STATE \textbf{Initialize:} \(R^{\text{best}}=P_{\text{c,r}}^\star=P_{\text{c,d}}^\star=0\) and 
\(\pmb{P}_{\text{p}}^\star=\pmb{C}^\star=\pmb{0}_{U}\).
\FOR{each $P_{\text{c,r}}\in\mathcal{P}_{\text{c,r}}$}
    \FOR{each $P_{\text{c,d}}\in\mathcal{P}_{\text{c,d}}$}
        \STATE Set $\tilde{P}_{\text{p},u}=\frac{MNP_{\text{max}}-P_{\text{c,r}}-N_{\text{c}}P_{\text{c,d}}}{MNU}$, $\forall u\in\mathcal{U}$.
        \STATE Calculate $\boldsymbol{\Sigma}_u$, $\kappa_u$, $\eta_{u,a}$, $R_c$ and $\rho_{u,a}$, $\forall u,a$.
        \STATE Initialize $R^{\text{old}}=0$ and $\delta=\infty$.
        \REPEAT
            \STATE Solve Problem~\eqref{original Problem_Taylor} using CVX.
            \IF{Problem~\eqref{original Problem_Taylor} is infeasible}
                \STATE Go to the next grid point.
            \ENDIF
            \STATE Compute $R^{\text{new}}=R^{\text{sum}}(\pmb{P}_{\text{p}},\pmb{C})$.
            \STATE Set $\delta=|R^{\text{new}}-R^{\text{old}}|/R^{\text{new}}$.
            \STATE Set $\tilde{P}_{\text{p},u}\leftarrow P_{\text{p},u}$, $\forall u$, and $R^{\text{old}}\leftarrow R^{\text{new}}$.
        \UNTIL{$\delta\le\epsilon$}
        \IF{$R^{\text{new}}>R^{\text{best}}$}
            \STATE Update $R^{\text{best}}\leftarrow R^{\text{new}}$, $P_{\text{c,r}}^\star\leftarrow P_{\text{c,r}}$, $P_{\text{c,d}}^\star\leftarrow P_{\text{c,d}}$, $\pmb{P}_{\text{p}}^\star\leftarrow\pmb{P}_{\text{p}}$, and $\pmb{C}^\star\leftarrow\pmb{C}$.
        \ENDIF
    \ENDFOR
\ENDFOR
\STATE \textbf{Output:} $P_{\text{c,r}}^\star$, $P_{\text{c,d}}^\star$, $\pmb{P}_{\text{p}}^\star$, and $\pmb{C}^\star$.
\end{algorithmic}
\end{algorithm}

We next analyze the computational complexity of \textbf{SCA-2D}.
At each 2D grid point, computing the surrogate-rate coefficients requires
\(\mathcal{O}(UMNQ)\) operations, while each SCA iteration requires
\(\mathcal{O}(UMN)\) operations for Taylor linearization and
\(\mathcal{O}((2U)^{3.5})\) for solving the convex subproblem.
Thus, the overall complexity is
\(
\mathcal{O}\left(
\frac{1}{\epsilon_r\epsilon_d}
\left[
UMNQ+
I_{\text{SCA}}\left(UMN+(2U)^{3.5}\right)
\right]
\right)\),
where $\epsilon_r$, $\epsilon_d$, and \(I_{\text{SCA}}\) denote the
pilot-power grid resolution, common-data-power grid resolution, and number
of SCA iterations, respectively.

\subsection{Low-Complexity Optimization Algorithm}
To further reduce the computational complexity of \textbf{SCA-2D}, we
propose a low-complexity method, referred to as \textbf{L-2D}, inspired
by~\cite{RSMA_power}. Like \textbf{SCA-2D}, \textbf{L-2D} performs a
2D search over \(P_{\text{c,r}}\) and \(P_{\text{c,d}}\). However, instead
of iterative convex optimization at each grid point, \textbf{L-2D}
determines the private-power and common-rate allocations analytically
using a conservative lower-bound approximation of the private-message rate.
The detailed procedure is presented below.

We first consider a fixed feasible pair
\((P_{\text{c,r}},P_{\text{c,d}})\). For this fixed pair,
Lemma~\ref{Lemma2} determines the total private-message power
\(P_{\text{p}}\), and the allocable common rate is given by
\(R_{\text{c}}=\min_{u}R_{\text{c},u}^{\text{s}}\).
To obtain a tractable lower bound on the surrogate private rate, we retain
only the common-data positions and upper-bound their effective interference,
as follows: 
\begin{align}
\label{surrogate private rate}
R_{\text{p},u}^{\text{s}}
&=
\frac{1}{MN}
\sum_{a=1}^{MN}
\log_2\!\left(
\frac{\eta_{u,a}}{\eta_{u,a}-\kappa_u P_{\text{p},u}}
\right) \notag \\
&\overset{(a)}{\geq}
\frac{1}{MN}
\sum_{a\in\mathcal{A}_{\text{data}}}
\log_2\!\left(
\frac{\eta_{u,a}}{\eta_{u,a}-\kappa_u P_{\text{p},u}}
\right) \notag \\
&\overset{(b)}{\geq}
\frac{N_{\text{c}}}{MN}
\log_2\!\left(
\frac{\eta_u}{\eta_u-\kappa_u P_{\text{p},u}}
\right)
=
\bar{R}_{\text{p},u}^{\text{s}},
\end{align}
where $\eta_u =
P_{\text{c,d}}\operatorname{Tr}(\mathbf{\Sigma}_u)
+\sigma_u^2P_{\text{p}}
+\sigma_n^2$. 
Here, \((a)\) follows by dropping the nonnegative terms outside
\(\mathcal{A}_{\text{data}}\).
For \((b)\), the guard region excludes the pilot position from
\(\mathcal{I}_u(a)\) for \(a\in\mathcal{A}_{\text{data}}\), so that
\(P_{\text{c}}(b)\leq P_{\text{c,d}}\) for all
\(b\in\mathcal{I}_u(a)\).
Using $\sum_{b\in\mathcal{I}_u(a)}
\Sigma_u^{\text{err}}(a,b)
=
\operatorname{Tr}(\mathbf{\Sigma}_u)$,  
we obtain \(\eta_{u,a}\leq\eta_u\).
Since
\(\log_2\!\bigl(\eta/(\eta-\kappa_uP_{\text{p},u})\bigr)\)
is decreasing in \(\eta\), this proves \((b)\). Thus, \(\bar{R}_{\text{p},u}^{\text{s}}\) provides a tractable lower bound
on \(R_{\text{p},u}^{\text{s}}\). Replacing
\(R_{\text{p},u}^{\text{s}}(P_{\text{p},u})\) in
\eqref{common message, pilot given problem} with
\(\bar{R}_{\text{p},u}^{\text{s}}\), we obtain the following problem:
\begin{subequations}
\label{reformulated_problem}
\begin{align}
\max_{\pmb{P}_{\text{p}},\pmb{C}} &\quad 
R_{\text{c}}+ \sum_u \bar{R}_{\text{p},u}^{s}({P}_{\text{p},u}) \label{L-2D_objective}\\
 \text{s.t.} \quad &  C_u +\bar{R}_{\text{p},u}^{s}({P}_{\text{p},u}) \ge R^{\text{th}}, \quad \forall u, \label{rate_threshold2} \\
& \sum_u C_u=  R_{\text{c}},  \label{common_mssage_allocaiton_equal}\\
& \sum_u P_{\text{p},u} =  {P_{\text{p}}}, \label{power_constraint_equal} \\
& \eqref{Common,privatepower_positive_const}. 
\end{align}
\end{subequations}

To derive the allocation structure, we first focus on the case
$R_{\text{c}}<UR^{\text{th}}$, where the common rate alone is
insufficient to satisfy all users' rate requirements. We first establish
the following lemma.

\begin{lemma}
\label{lem:common_rate_bound}
For \(R_{\text{c}}<UR^{\text{th}}\), Problem~\eqref{reformulated_problem}
can be restricted to \(0\le C_u\le R^{\text{th}}\),
\(\forall u\in\mathcal U\), without loss of optimality.
\end{lemma}

\begin{proof}
Consider an optimal solution of Problem~\eqref{reformulated_problem}. If it
already satisfies \(C_u\le R^{\text{th}}\) for all \(u\in\mathcal U\), the
claim follows. Otherwise, suppose that \(C_u>R^{\text{th}}\) for some user
\(u\). Since \(\sum_{w\in\mathcal U}C_w=R_{\text{c}}<UR^{\text{th}}\),
there exists a user \(v\neq u\) with \(C_v<R^{\text{th}}\). Shifting common
rate from \(u\) to \(v\) until either \(C_u=R^{\text{th}}\) or
\(C_v=R^{\text{th}}\) preserves the total common rate, private-message powers,
and objective value, without violating any rate constraints. Repeating this
operation yields an optimal solution with
\(0\le C_u\le R^{\text{th}}\), \(\forall u\in\mathcal U\).
\end{proof}
For a given \(0\le C_u\le R^{\text{th}}\), the minimum private-message power required for
user \(u\) to satisfy \(R^{\text{th}}\) is
$P_{\text{p},u}^{\text{min}}(C_u)
=\frac{1-\phi_u(C_u)}{\lambda_u}$,
where $
\lambda_u=\frac{\kappa_u}{\eta_u}$ and $
\phi_u(C_u)=2^{-MN(R^{\text{th}}-C_u)/N_{\text{c}}}$. 
The optimal allocations for Problem~\eqref{reformulated_problem}
are then characterized in the following theorem.

\begin{theorem}
\label{theorem2}
For a fixed candidate pair \((P_{\text{c,r}},P_{\text{c,d}})\) such that 
Problem~\eqref{reformulated_problem} is feasible and
\(0\le R_{\text{c}}<UR^{\text{th}}\), let
\((\pi_1,\pi_2,\ldots,\pi_U)\) be a permutation of \(\mathcal U\) such that
\(\lambda_{\pi_1}\le\lambda_{\pi_2}\le\cdots\le\lambda_{\pi_U}\).
Then, an optimal common-rate allocation is given by
\begin{align}
\label{common rate allocation optimal}
C_{\pi_j}^{\star}
=
\min\left\{
R^{\text{th}},
\left[R_{\text{c}}-(j-1)R^{\text{th}}\right]^+
\right\},
\quad \quad j=1,\ldots,U,
\end{align}
where $[x]^+=\max\{x,0\}$. The corresponding optimal private-power allocation is given by
\begin{align}
\label{optimal_joint_private_power}
P_{\text{p},\pi_j}^{\star}
&=
P_{\text{p},\pi_j}^{\text{min}}(C_{\pi_j}^{\star}),
\quad j=1,\ldots,U-1, \notag\\
P_{\text{p},\pi_U}^{\star}
&=
P_{\text{p}}
-
\sum_{j=1}^{U-1}
P_{\text{p},\pi_j}^{\text{min}}(C_{\pi_j}^{\star}).
\end{align}
\end{theorem}

\begin{proof}
We first optimize the private-message powers for a fixed feasible
common-rate allocation $\pmb{C}$ satisfying
$0\le C_u\le R^{\text{th}}$.
Define $\phi(C)=2^{-MN(R^{\text{th}}-C)/N_{\text{c}}}$ and write
$\phi_u=\phi(C_u)$ and
$P_{\text{p},u}^{\text{min}}
=P_{\text{p},u}^{\text{min}}(C_u)$ for brevity.
The private-power constraints become
$P_{\text{p},u}\ge P_{\text{p},u}^{\text{min}}$ and
$\sum_u P_{\text{p},u}=P_{\text{p}}$.
Let
\(
P^{\text{rem}}(\pmb{C})
=
P_{\text{p}}-\sum_{u\in\mathcal U}P_{\text{p},u}^{\text{min}}
\ge0.
\)
Since the sum private rate is convex over this feasible polytope,
an optimal solution is attained at an extreme
point~\cite{boyd2004convex}.
Thus, all remaining private-message power can be assigned
to a single user $u'$, yielding
\begin{align}
P_{\text{p},u}^{\star}
=P_{\text{p},u}^{\text{min}},
\quad u\ne u', \quad
P_{\text{p},u'}^{\star}
=P_{\text{p},u'}^{\text{min}}
+P^{\text{rem}}(\pmb{C}).
\label{optimal private message power}
\end{align}
Using
$1-\lambda_u P_{\text{p},u}^{\text{min}}=\phi_u$
and
$-(N_{\text{c}}/MN)\log_2\phi_u=R^{\text{th}}-C_u$,
the resulting sum private rate is
\begin{align}
\sum_{u\in\mathcal U}\bar R_{\text{p},u}^{\text{s}}
=
UR^{\text{th}}-R_{\text{c}}
-\frac{N_{\text{c}}}{MN}
\log_2\!\left(
1-\frac{\lambda_{u'}}{\phi_{u'}}
P^{\text{rem}}(\pmb{C})
\right).
\label{sum_private_rate}
\end{align}
For fixed $\pmb{C}$, $P^{\text{rem}}(\pmb{C})$ is nonnegative
and independent of $u'$.
Since $-\log_2(1-x)$ is increasing in $x$,
one optimal choice is
\begin{align}
u'=
\arg\max_{u\in\mathcal U}
\frac{\lambda_u}{\phi(C_u)}.
\label{user selcection}
\end{align}
Define
\[
\Psi(\pmb{C})
=
\max_{u\in\mathcal U}\frac{\lambda_u}{\phi(C_u)},
\qquad
\Gamma(\pmb{C})
=
P^{\text{rem}}(\pmb{C})\Psi(\pmb{C}).
\]
By \eqref{user selcection} and $\phi_{u'}=\phi(C_{u'})$,
the selected user satisfies
$\lambda_{u'}/\phi_{u'}=\Psi(\pmb{C})$.
Thus,
$(\lambda_{u'}/\phi_{u'})P^{\text{rem}}(\pmb{C})
=\Gamma(\pmb{C})$.
Substituting this identity into \eqref{sum_private_rate},
and noting that $R_{\text{c}}$ is fixed, shows that
the remaining common-rate optimization is equivalent to
maximizing $\Gamma(\pmb{C})$ over the feasible
common-rate allocations.

Let $\pmb{C}^{\star}$ denote the candidate allocation in
\eqref{common rate allocation optimal}.
We next show that it maximizes both
$P^{\text{rem}}(\pmb{C})$ and $\Psi(\pmb{C})$.

First, observe that
\(P^{\text{rem}}(\pmb{C})
=
P_{\text{p}}
-\sum_{u\in\mathcal U}\frac{1}{\lambda_u}
+\sum_{u\in\mathcal U}\frac{\phi(C_u)}{\lambda_u}\). 
Consider maximizing this expression under only
\(\sum_u C_u=R_{\text{c}}\) and
\(0\le C_u\le R^{\text{th}}\), temporarily omitting
the private-power feasibility condition
\(P^{\text{rem}}(\pmb{C})\ge0\).
Since \(\phi\) is convex, a maximizer over this polytope
exists at an extreme point, where at most one component
lies strictly between \(0\) and \(R^{\text{th}}\).
Hence, letting
\(L=\lfloor R_{\text{c}}/R^{\text{th}}\rfloor\) and
\(r=R_{\text{c}}-LR^{\text{th}}\), the components of an
extreme-point maximizer satisfy, up to permutation,
\begin{align}
\label{extreme_common_rate_structure}
\{C_u\}_{u\in\mathcal U}
=
\{
\underbrace{R^{\text{th}},\ldots,R^{\text{th}}}_{L},
r,
\underbrace{0,\ldots,0}_{U-L-1}
\}.
\end{align}

Next, suppose that a feasible common-rate allocation \(\pmb{C}\) does not
satisfy the desired ordering, i.e.,
\(C_{\pi_i}<C_{\pi_j}\) for some \(i<j\).
Let \(\pmb{C}'\) be obtained by swapping these two entries. Then,
$P^{\text{rem}}(\pmb{C}')
\!-\!P^{\text{rem}}(\pmb{C})
\!=\!
\left(
\frac{1}{\lambda_{\pi_i}}
\!-\!\frac{1}{\lambda_{\pi_j}}
\right)
\left(
\phi(C_{\pi_j})\!-\!\phi(C_{\pi_i})
\right)
\ge0$, 
where the inequality follows from
\(\lambda_{\pi_i}\le\lambda_{\pi_j}\) and the monotonicity of \(\phi\).
Thus, the swap preserves feasibility while not decreasing the objective.
Repeating such swaps yields an ordered allocation with no smaller objective.
Therefore, there exists an optimal allocation satisfying
\(C_{\pi_1}\ge C_{\pi_2}\ge\cdots\ge C_{\pi_U}\).
Combining this ordering with the extreme-point structure in
\eqref{extreme_common_rate_structure} yields exactly
\(\pmb{C}^{\star}\) in \eqref{common rate allocation optimal}.

Second, we show that \(\pmb{C}^{\star}\) also maximizes
\(\Psi(\pmb{C})\). From the sum and box constraints, any feasible
common-rate allocation satisfies
\begin{align}
\label{common_rate_lower_bound}
C_u
\ge
\bigl[R_{\text{c}}-(U-1)R^{\text{th}}\bigr]^+
=
C_{\pi_U}^{\star},
\qquad \forall u\in\mathcal U,
\end{align}
where the equality follows from
\eqref{common rate allocation optimal} and
\(R_{\text{c}}<UR^{\text{th}}\).
Thus, \(C_{\pi_U}^{\star}\) is a common lower bound on the common-rate
portion of every user.
Since $\lambda_u\le\lambda_{\pi_U}$ and $\phi$ is increasing,
\[
\Psi(\pmb{C})
=
\max_u\frac{\lambda_u}{\phi(C_u)}
\le
\frac{\lambda_{\pi_U}}{\phi(C_{\pi_U}^{\star})}
=
\Psi(\pmb{C}^{\star}),
\]
where the upper bound is attained by user $\pi_U$
under $\pmb{C}^{\star}$.

Since both factors are nonnegative for feasible allocations,
the preceding inequalities yield
\(
\Gamma(\pmb{C})
\le
P^{\text{rem}}(\pmb{C}^{\star})
\Psi(\pmb{C}^{\star})
=
\Gamma(\pmb{C}^{\star}).
\)
Hence, $\pmb{C}^{\star}$ is an optimal common-rate allocation.
Moreover, since
$\lambda_{\pi_U}/\phi(C_{\pi_U}^{\star})
=\Psi(\pmb{C}^{\star})$,
user $\pi_U$ satisfies the selection rule in
\eqref{user selcection} under $\pmb{C}^{\star}$.
Hence, one can choose $u'=\pi_U$.
Substituting $u'=\pi_U$ into
\eqref{optimal private message power} gives the stated
private-power allocation.
\end{proof}
We next consider the complementary case $R_{\text{c}}\ge UR^{\text{th}}$,
where the common rate alone can satisfy all users' rate requirements.

\begin{theorem}
\label{theorem3}
For a fixed feasible candidate pair $(P_{\text{c,r}},P_{\text{c,d}})$
with $R_{\text{c}}\ge UR^{\text{th}}$, let $\pi_U$ be a user with the
largest $\lambda_u$. Then, one optimal solution of
Problem~\eqref{reformulated_problem} is given by
\begin{align}
C_u^{\star}=\frac{R_{\text{c}}}{U}, \quad \forall u\in\mathcal U, \quad P_{\text{p},\pi_U}^{\star}=P_{\text{p}},\quad
P_{\text{p},u}^{\star}=0, \quad \forall u\ne\pi_U.
\notag
\end{align}
\end{theorem}

\begin{proof}
Since \(C_u^{\star}=R_{\text{c}}/U\ge R^{\text{th}}\), all users satisfy the
rate threshold through the common rate alone. Thus, the private-message power
allocation reduces to maximizing the sum private rate over the simplex
\(\sum_uP_{\text{p},u}=P_{\text{p}}\), \(P_{\text{p},u}\ge0\). Since this objective is convex over the simplex, an extreme-point solution is
optimal, with all private-message power assigned to a single user \(u\). The resulting sum private rate then reduces to
\(-\frac{N_{\text{c}}}{MN}\log_2(1-\lambda_uP_{\text{p}})\), which is
nondecreasing in \(\lambda_u\). Hence, choosing a user \(\pi_U\) with the
largest \(\lambda_u\) gives one optimal solution.
\end{proof}
The optimal allocations in Theorems~\ref{theorem2} and~\ref{theorem3}
require ordering the users according to \(\lambda_u\).
Since \(\lambda_u\) depends on \(P_{\text{c,r}}\) and \(P_{\text{c,d}}\),
this ordering may vary across grid points.
To avoid such grid-dependent reordering, we introduce the following lemma
based on long-term channel statistics.
\begin{lemma}
\label{Lemma3}
Let \((\tilde{\pi}_1,\ldots,\tilde{\pi}_U)\) denote the user ordering
according to the long-term channel powers \(\{\sigma_u^2\}\), i.e., 
$\sigma_{\tilde{\pi}_1}^2
\le \sigma_{\tilde{\pi}_2}^2
\le \cdots
\le \sigma_{\tilde{\pi}_U}^2$.
Under the high-pilot-SNR condition\footnote{This condition is consistent with the operating regime of
interest, since reliable data detection requires sufficiently accurate
channel estimation.}
$P_{\text{c,r}}\sigma_{u,q}^2
\gg \tilde{\sigma}_{n,u}^2,
\forall u,q$,
\(\lambda_u\) is approximately increasing in \(\sigma_u^2\).
Hence, the ordering induced by \(\{\lambda_u\}\) can be approximated by
the long-term channel-power ordering
\((\tilde{\pi}_1,\ldots,\tilde{\pi}_U)\).
In particular, \(\tilde{\pi}_U\) approximately identifies the user with
the largest \(\lambda_u\).
\end{lemma}
\begin{proof}
When $ P_{\text{c,r}}\sigma_{u,q}^2 \gg \tilde{\sigma}_{n,u}^2 $ for all $q$,
\(
\operatorname{Tr}(\boldsymbol{\Sigma}_{u})
=
\sum_q
\frac{\tilde{\sigma}_{n,u}^2 \sigma_{u,q}^2}{P_{\text{c,r}}\sigma_{u,q}^2+\tilde{\sigma}_{n,u}^2}
\approx
Q\frac{\tilde{\sigma}_{n,u}^2}{P_{\text{c,r}}}
=
Q\frac{\sigma_u^2 P_{\text{p}}+\sigma_n^2}{P_{\text{c,r}}}.
\)
As a result, the ratio $ \lambda_u=\frac{\kappa_u}{\eta_u} $ can be approximated as 
\[
\lambda_u \!=\!
\frac{
\sigma_{u}^2 - \operatorname{Tr}(\boldsymbol{\Sigma}_{u})
}{
\operatorname{Tr}(\boldsymbol{\Sigma}_{u}) P_{\text{c,d
}}
\!+\! \tilde{\sigma}_{n,u}^2
} \!\approx \!\frac{
\sigma_{u}^2 (1- Q\frac{P_{\text{p}}}{P_{\text{c,r}}})-Q\frac{\sigma_n^2}{P_{\text{c},r}}}{
\sigma_{u}^2 (P_{\text{p}}+Q\frac{P_{\text{p}}P_{\text{c,d}}}{P_{\text{c,r}}}) \!+\! \sigma_n^2+Q\frac{\sigma_n^2P_{\text{c,d}}}{P_{\text{c,r}}}
}.
\]
It can be readily verified by differentiation that the above approximation of
\(\lambda_u\) is strictly increasing with respect to \(\sigma_u^2\). Hence,
under the high-pilot-SNR approximation, the ordering of \(\lambda_u\) follows
the ordering of \(\sigma_u^2\). 
\end{proof}

Together with Lemma~\ref{Lemma3}, Theorems~\ref{theorem2}
and~\ref{theorem3} reveal a simple allocation structure based on the
long-term channel-gain ordering. 
When \(R_{\text{c}}<UR^{\text{th}}\), the common rate is sequentially
allocated from weaker to stronger users, up to \(R^{\text{th}}\) per user.
Each user is then assigned the minimum private-message power required to meet \(R^{\text{th}}\), with the remaining private-message power allocated to the strongest user.
When \(R_{\text{c}}\ge UR^{\text{th}}\), the common rate is equally divided among the users, and all private-message power is allocated to the strongest user.

Based on this allocation structure,
\textbf{L-2D} first orders the users according to their long-term channel
gains and uses this ordering throughout the two-dimensional search.
For each candidate pair \((P_{\text{c,r}},P_{\text{c,d}})\),
\(R_{\text{c}}\) is computed, and the common-rate and private-message power
allocations are determined depending on whether
\(R_{\text{c}}<UR^{\text{th}}\) or \(R_{\text{c}}\ge UR^{\text{th}}\).
The candidate pair that maximizes the objective in
\eqref{L-2D_objective} is then selected. 
Since \textbf{L-2D} follows the same two-dimensional search procedure as
\textbf{SCA-2D}, with the iterative SCA step replaced by the closed-form
allocation rules, its pseudocode is omitted for brevity.

We next analyze the computational complexity of \textbf{L-2D}.
The dominant per-grid-point cost is the computation of the
position-dependent surrogate-rate coefficients, with complexity
\(\mathcal{O}(UMNQ)\).
Hence, the overall complexity is
\(\mathcal{O}\left(\frac{UMNQ}{\epsilon_d\epsilon_r}\right).
\)
Compared with \textbf{SCA-2D}, \textbf{L-2D} eliminates the iterative
SCA optimization at each grid point, substantially reducing the overall
computational complexity.
\begin{remark}
Although the proposed power-allocation algorithms are developed for the
GS configuration, they can be extended to GG, SG, and SS by updating the
configuration-dependent channel-estimation error terms and rate expressions.
These modifications affect only the corresponding coefficients, while
preserving the optimization structure; hence, the closed-form allocation
structures in Theorems~2 and~3 remain valid. Likewise, when the high-pilot-SNR condition holds for each configuration,
Lemma~4 allows the user ordering to be approximated by the same long-term
channel-power ordering. Thus, the resulting structural
and physical interpretations apply broadly across the considered OTFS-RSMA
configurations.
\end{remark}


\begin{table*}[t]
\centering
\caption{Optimized power allocation for \textbf{SCA-2D} and \textbf{L-2D} at different SNR values, with $P_{\text{max}}=1$ W and $R^{\text{th}}=0.5$ bps/Hz.}
\label{tab:power_alloc_bis_tay}
\begin{tabular}{c|ccccc|ccccc}
\hline
\multirow{2}{*}{SNR (dB)} 
& \multicolumn{5}{c|}{\textbf{SCA-2D}} 
& \multicolumn{5}{c}{\textbf{L-2D}} \\
& $P_{\text{c,d}}$ (W) & $P_{\text{c,r}}$ (W) & $P_{\text{p},1}$ (W) & $P_{\text{p},2}$ (W) & $P_{\text{p},3}$ (W)
& $P_{\text{c,d}}$ (W) & $P_{\text{c,r}}$ (W) & $P_{\text{p},1}$ (W) & $P_{\text{p},2}$ (W) & $P_{\text{p},3}$ (W) \\
\hline
20 & 0.6 & 395 & 0.2871 & 0 & 0 & 0.7 & 370 & 0.2126 & 0 & 0 \\
22 & 0.6 & 475 & 0.2480 & 0 & 0 & 0.8 & 355 & 0.1333 & 0 & 0 \\
24 & 0.7 & 475 & 0.1614 & 0 & 0 & 0.8 & 420 & 0.1016 & 0 & 0 \\
26 & 0.8 & 435 & 0.0942 & 0 & 0 & 0.9 & 345 & 0.0515 & 0 & 0 \\
\hline
\end{tabular}
\end{table*}

\begin{figure}[t]
    \centering
    \subfloat[]{
        \includegraphics[width=0.23\textwidth]{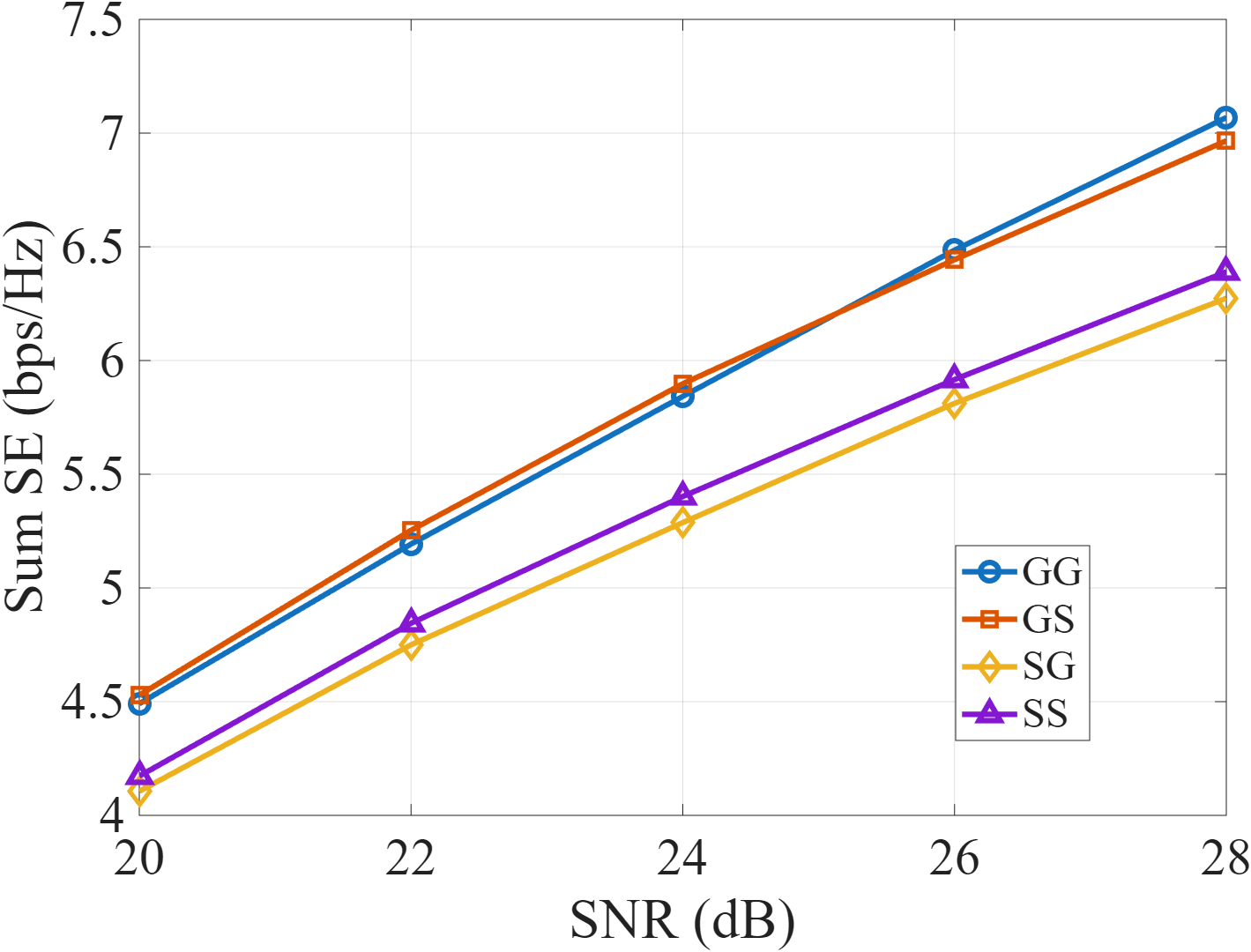}%
        \label{fig: GS_SNRvar}%
    }\hfill
    \subfloat[]{
        \includegraphics[width=0.22\textwidth]{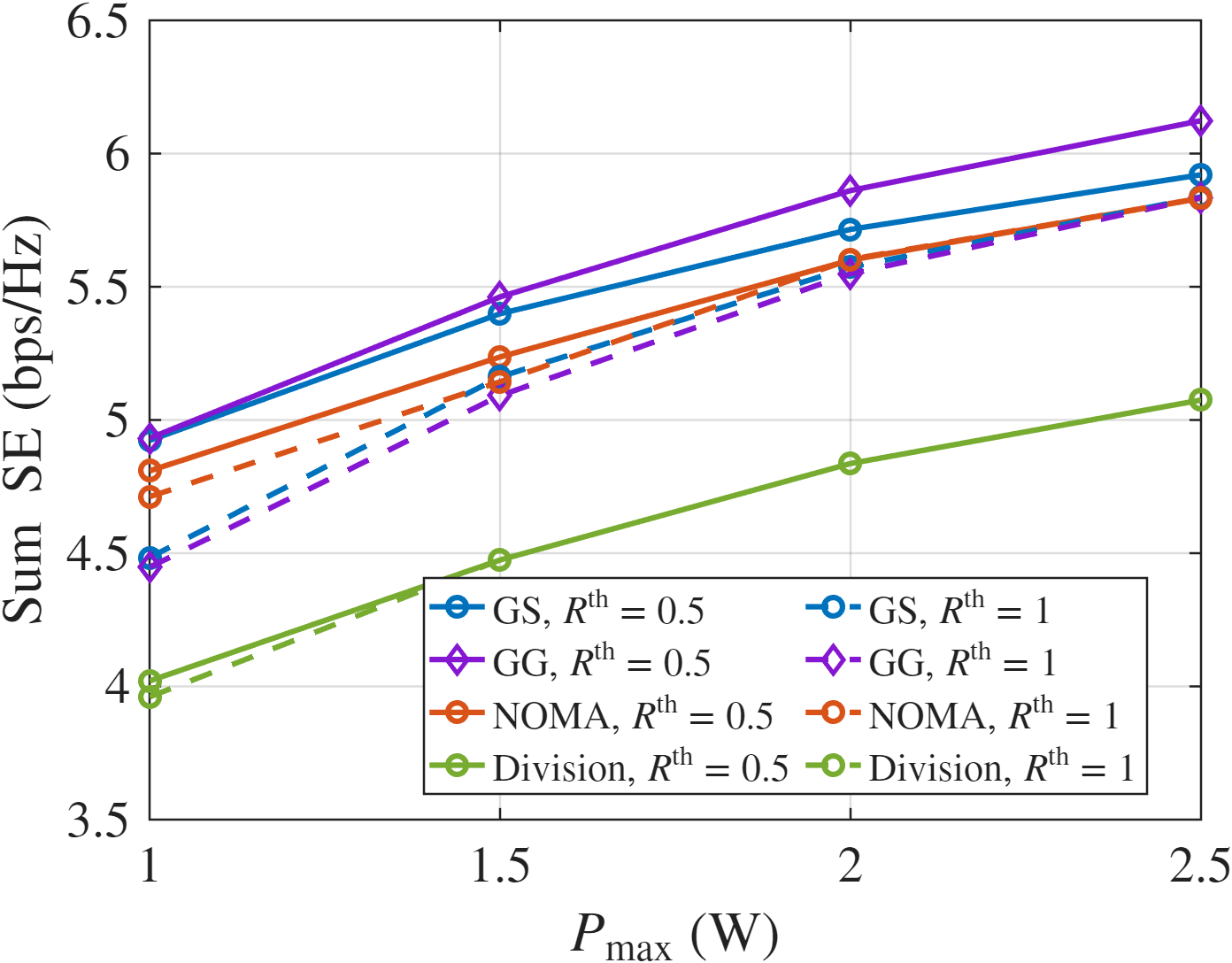}%
        \label{fig:RSMA_NOMA_divison_rate_compare}%
    }
\caption{Sum-SE performance: (a) comparison of pilot transmission
configurations versus SNR with $P_{\max}=1$~W and
$R^{\mathrm{th}}=1$~bps/Hz; (b) comparison of OTFS-RSMA
with OTFS-NOMA and orthogonal division at SNR $=20$~dB.}
    \label{GS_compare}
\end{figure}
 \begin{figure}[t]
    \centering
    \includegraphics[width=4cm]{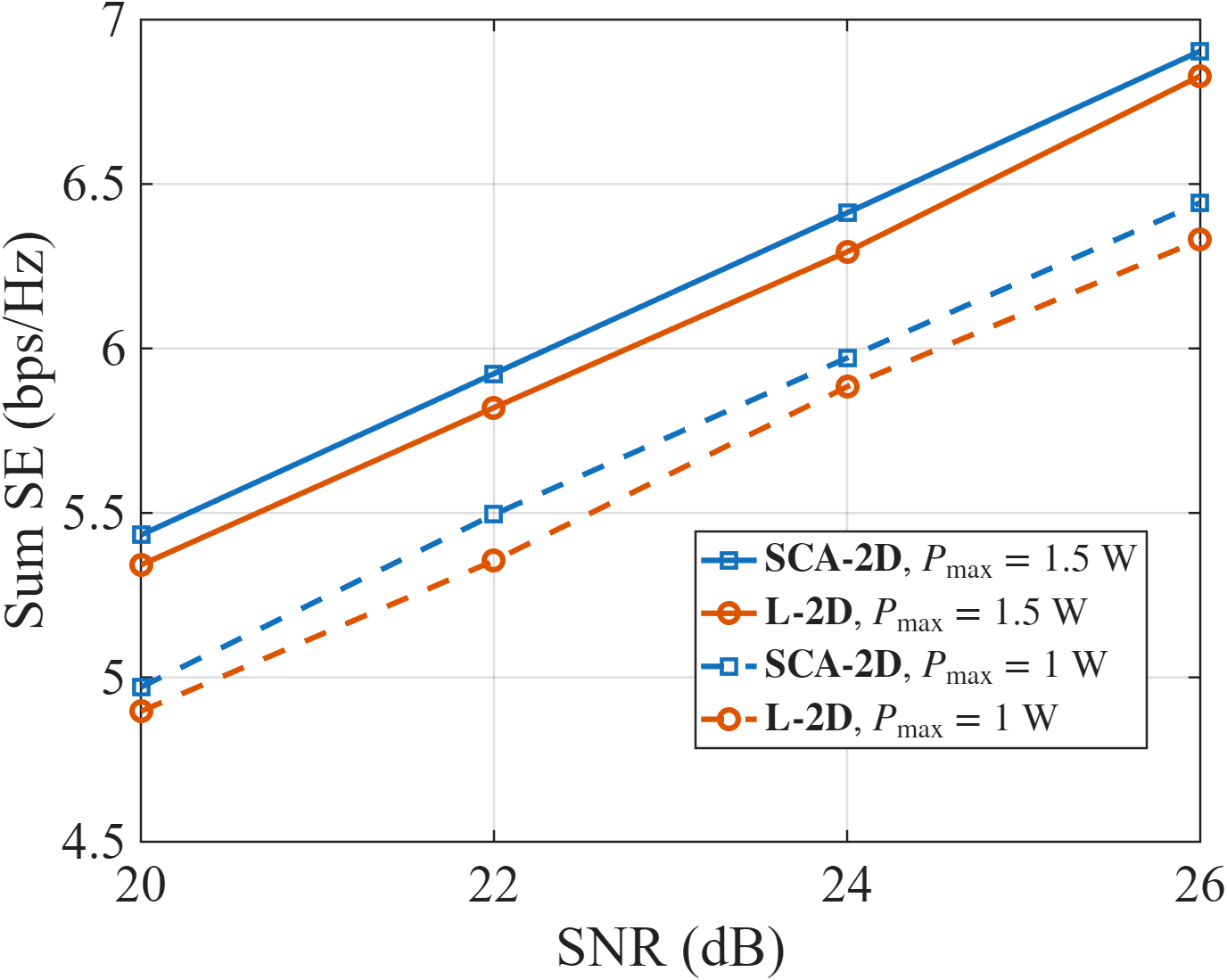}
\caption{{Sum SE versus SNR for SCA-2D and L-2D under the GS configuration
with $R^{\mathrm{th}}=0.5$ bps/Hz.}}
    \label{fig:Algorithm1 Algorithm2 compare graph}
\end{figure}
\section{Numerical Results}
We present numerical results to evaluate the proposed framework.
Unless otherwise specified, we consider \(U=3\) users, each with
\(Q=4\) propagation paths~\cite{OTFS_superimposed0}.
The overall channel powers of Users~2 and~3 are set to \(-3\) dB and
\(-9\) dB relative to that of User~1, respectively.
For each user, the variance of the path with the smallest delay is
normalized to \(0\) dB, and those of the remaining paths are determined
from their relative delays following
\cite{OTFS_channel_est1,OTFS_superimposed1}.
The OTFS frame consists of \(M=64\) delay bins and \(N=32\) Doppler bins,
with carrier frequency \(f_c=3.5\) GHz and subcarrier spacing
\(\Delta f=15\) kHz.
The maximum delay and Doppler indices are
\(l_{\max}=10\) (\(\tau_{\max}=10~\mu\text{s}\)) and
\(k_{\max}=3\) (\(v_{\max}=434~\text{km/h}\)), respectively
\cite{OTFS-NOMA}.
The noise variance is set as
\(\sigma_n^2=10^{-\text{SNR}/10}\); for example,
\(\text{SNR}=20\) dB gives \(\sigma_n^2=10^{-2}\).
The grid resolutions are \(\epsilon_d=0.1\) and \(\epsilon_r=5\), and
\(N_f=10^3\) Monte Carlo realizations are used for BER and channel-estimation
evaluation.
Because the surrogate rates used for optimization may differ from the
actual rates, we impose a 20\% margin on the minimum-rate requirement by
replacing \(R^{\text{th}}\) with \(1.2R^{\text{th}}\) during
optimization.
The reported \(R^{\text{th}}\) denotes the target rate without this margin, and its unit (bps/Hz) is omitted from the figure legends for brevity.


\subsection{Sum-SE Performance}
{We first compare the sum SE of the four pilot transmission configurations using \textbf{SCA-2D} 
in Fig.~\ref{fig: GS_SNRvar}.
Across the considered settings, GG and GS outperform SG and SS, indicating
that employing a guard-based structure for the common message is beneficial
by suppressing common-data interference during channel estimation.
Between GG and GS, however, neither configuration uniformly outperforms
the other.
GS allows the private messages to reuse the guard resources at the cost of
additional private-message interference during channel estimation, whereas
GG eliminates this interference by also guarding the private messages but
sacrifices the corresponding transmission resources.
Consequently, GS is preferable when the resource-utilization gain outweighs
the interference penalty, while GG becomes preferable as the impact of
private-message interference becomes more significant.} Fig.~\ref{fig:RSMA_NOMA_divison_rate_compare} compares the sum SE of
OTFS-RSMA using \textbf{SCA-2D} with those of the OTFS-NOMA~\cite{OTFS-NOMA} and
orthogonal-division~\cite{OTFS_channel_est1} benchmarks for different
\(P_{\max}\) and \(R^{\text{th}}\). For OTFS-RSMA, we consider the GG and GS
configurations because Fig.~\ref{fig: GS_SNRvar} shows that they achieve the
highest sum-SE performance among the four configurations, with their relative
performance depending on the operating conditions. In OTFS-NOMA, all users
share the entire DD grid through superposition, with the pilot and guard
symbols embedded in the weakest user's message, whereas orthogonal division
partitions the DD grid among users.
As shown in the figure, both RSMA and NOMA outperform orthogonal division,
highlighting the benefit of non-orthogonal resource sharing. RSMA outperforms
NOMA for smaller \(R^{\mathrm{th}}\) by serving weaker users through the common
message while concentrating private-message power on the strongest user,
whereas, for larger $R^{\mathrm{th}}$, NOMA achieves comparable
or slightly higher sum SE, benefiting from multi-stage SIC.
Within RSMA, GG is preferable for smaller \(R^{\mathrm{th}}\), where the larger private-message power can lead to stronger private-message interference during channel estimation. For larger \(R^{\mathrm{th}}\), the reduced private-message power mitigates this interference, making the spectral-efficiency benefit of GS more pronounced. Overall, RSMA achieves competitive performance relative to NOMA while requiring fewer SIC stages.

{We next compare \textbf{SCA-2D} and \textbf{L-2D} to evaluate the
performance--complexity tradeoff of \textbf{L-2D}.
Since the focus here is on the comparison between the two algorithms,
we present the results for the GS configuration, while similar trends
are observed for the other configurations.
As shown in Table~\ref{tab:power_alloc_bis_tay}, both algorithms yield
similar power-allocation patterns.
In particular, they allocate substantially more power to the pilot than
to the private messages, supporting the high-pilot-SNR assumption in
Lemma~\ref{Lemma3}.
Note that \(P_{\max}\) is the average power budget per DD resource element,
whereas \(P_{\text{c,r}}\) is the power of a single pilot symbol; hence,
\(P_{\text{c,r}}>P_{\max}\) does not violate the frame-average power
constraint.
Both algorithms also concentrate the private-message power on User~1,
the strongest user, consistent with the optimal allocation structure
characterized in Theorems~\ref{theorem2} and~\ref{theorem3}.
As shown in Fig.~\ref{fig:Algorithm1 Algorithm2 compare graph},
\textbf{L-2D} incurs less than a \(4\%\) actual sum-SE loss relative to
\textbf{SCA-2D}, while reducing the average computation time from
\(1382.5\) s to \(0.308\) s at \(P_{\max}=1.5\) W on an AMD Ryzen 9950X3D processor.

\begin{figure*}[t]
\centering
        \includegraphics[width=12.1cm]{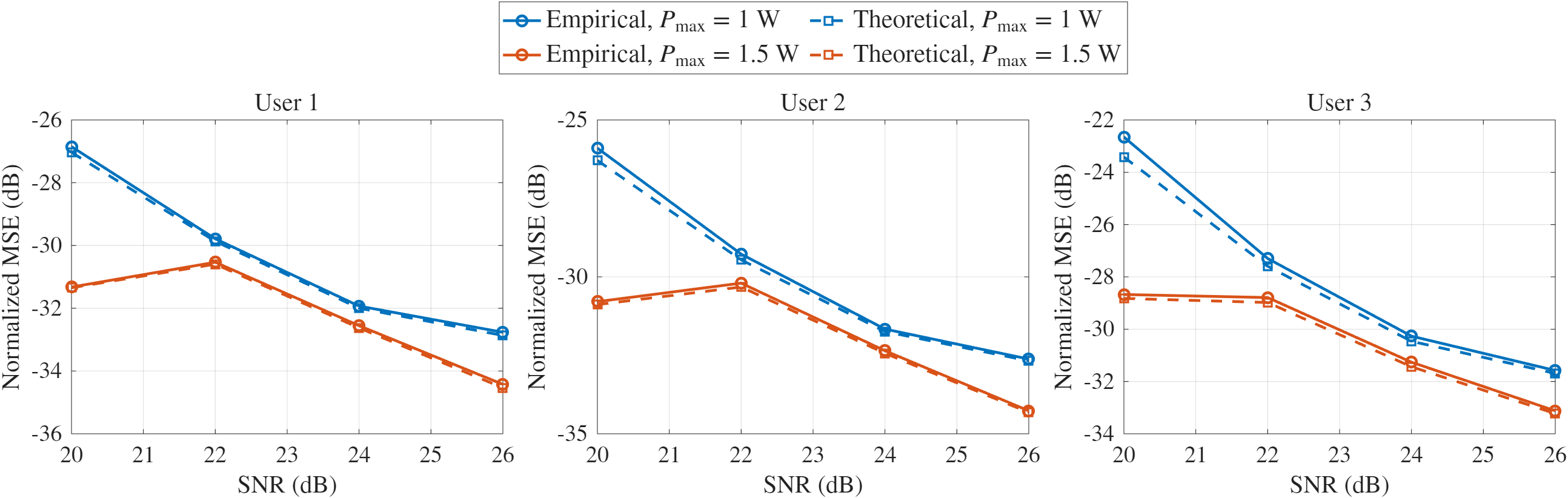}%
    \caption{Channel NMSE versus SNR for $P_{\max}=1$ and $1.5$~W.}
    \label{fig: channel_estmiatin}%
\end{figure*}
\begin{figure}[t]
    \centering
    \subfloat[]{
        \includegraphics[width=0.23\textwidth]{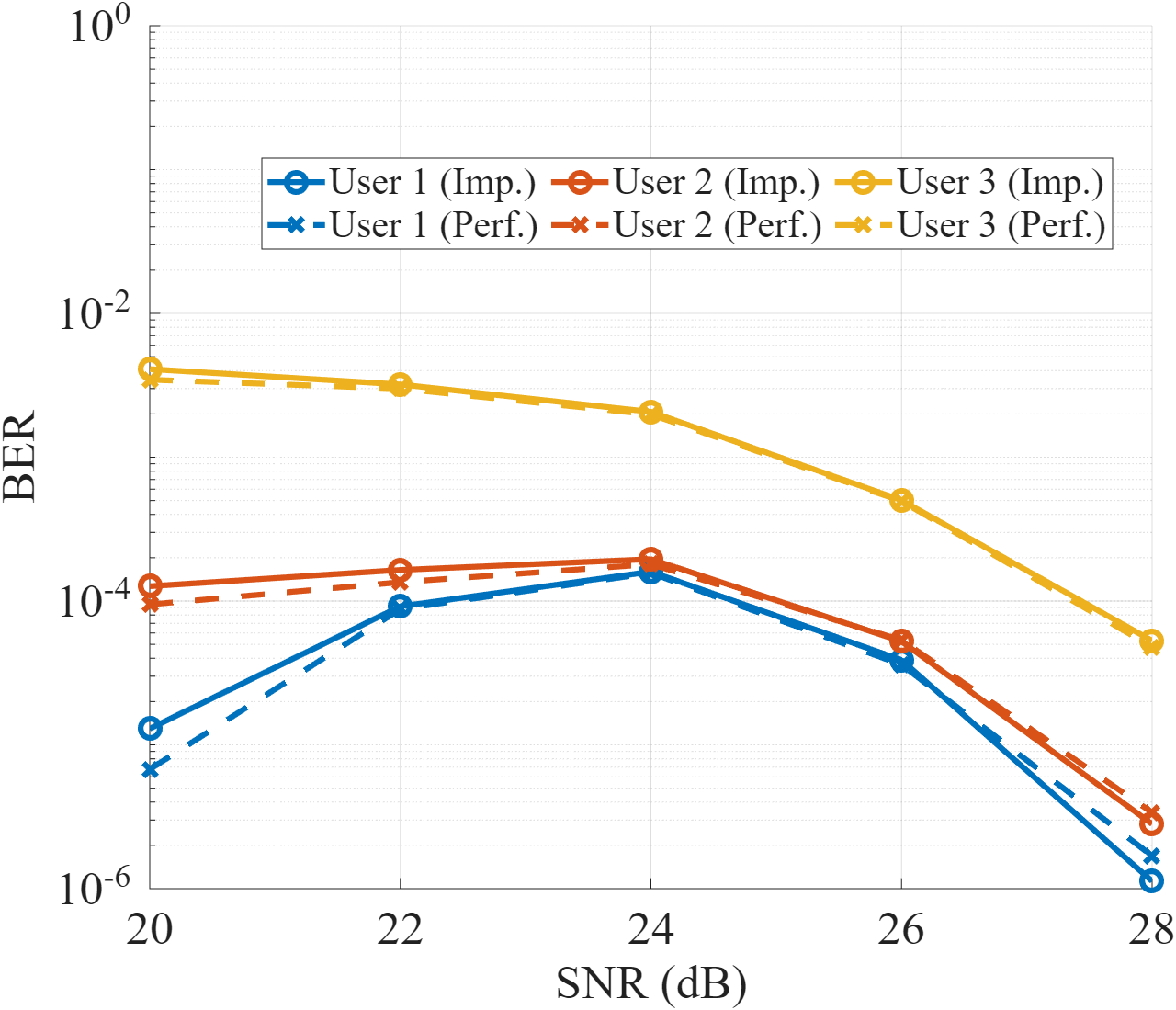}%
        \label{fig: Common Perfect_imperfect BER}%
    }\hfill
    \subfloat[]{
        \includegraphics[width=0.22\textwidth]{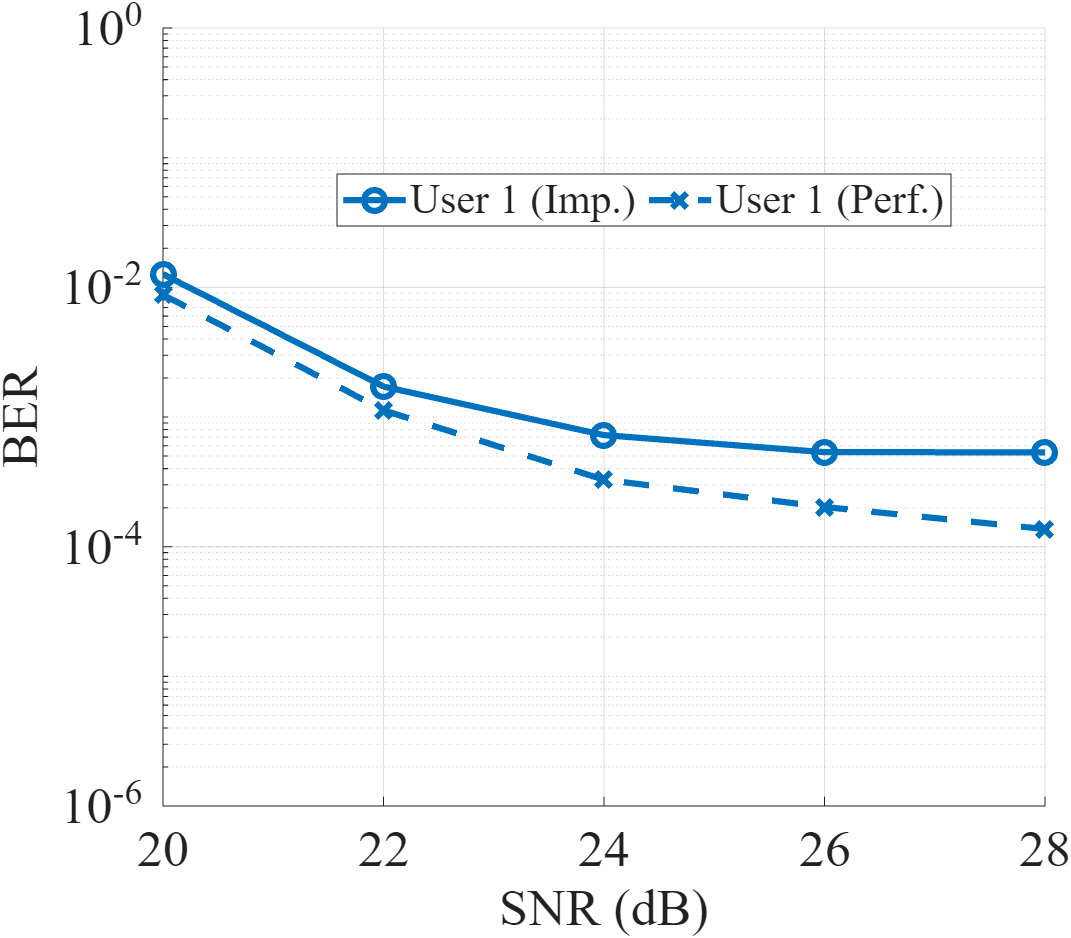}%
        \label{fig: Private Perfect_imperfect BER}%
    }
    \caption{BER  under perfect and imperfect CSI: (a) common data and (b) private message, with $P_{\text{max}}=1$ W and BPSK signaling.}
    \label{perfect_imperfect BER}
\end{figure}
\begin{figure}[t]
    \centering
    \subfloat[]{
        \includegraphics[width=0.23\textwidth]{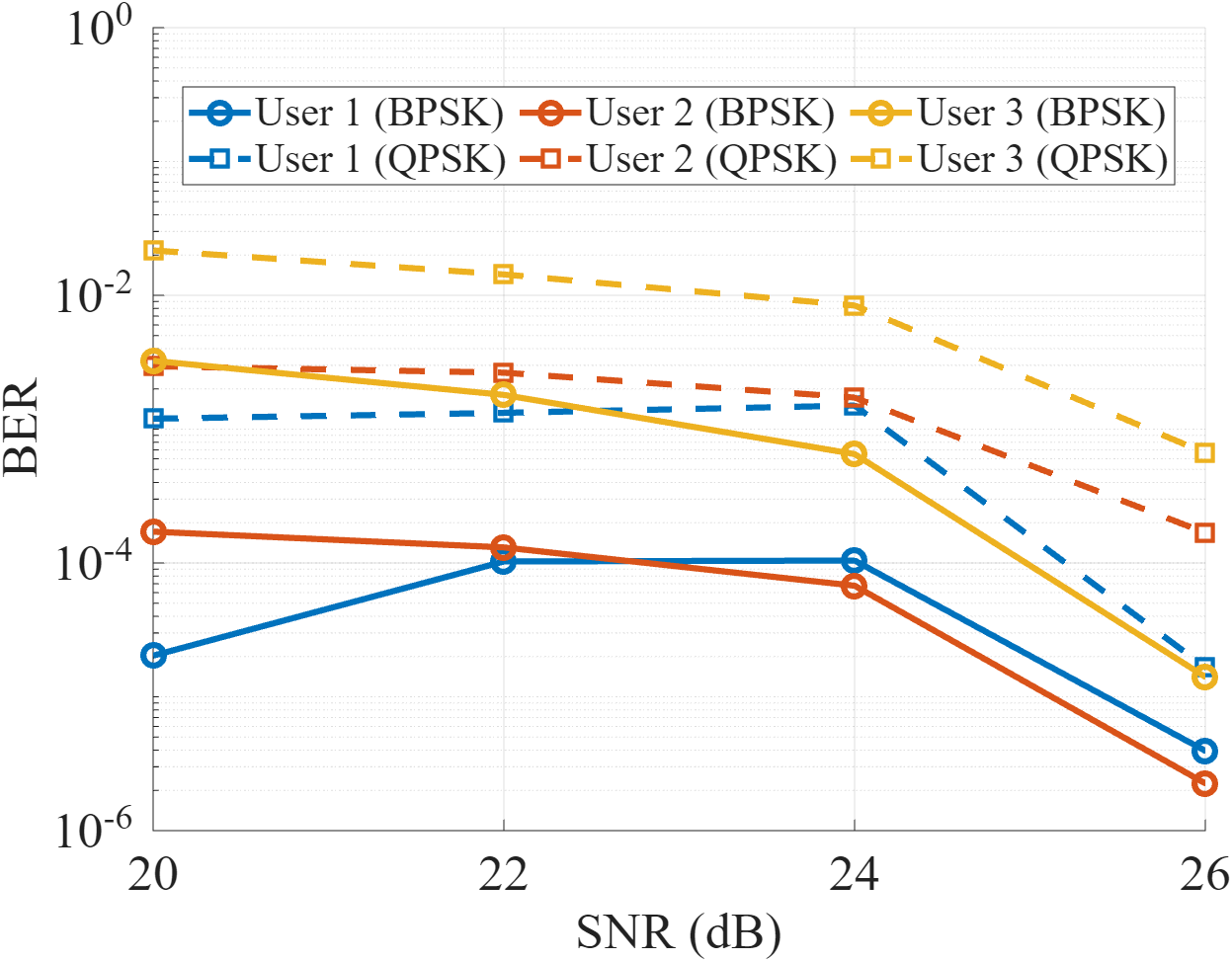}%
        \label{fig: Common BPSK QPSK BER}%
    }\hfill
    \subfloat[]{
        \includegraphics[width=0.22\textwidth]{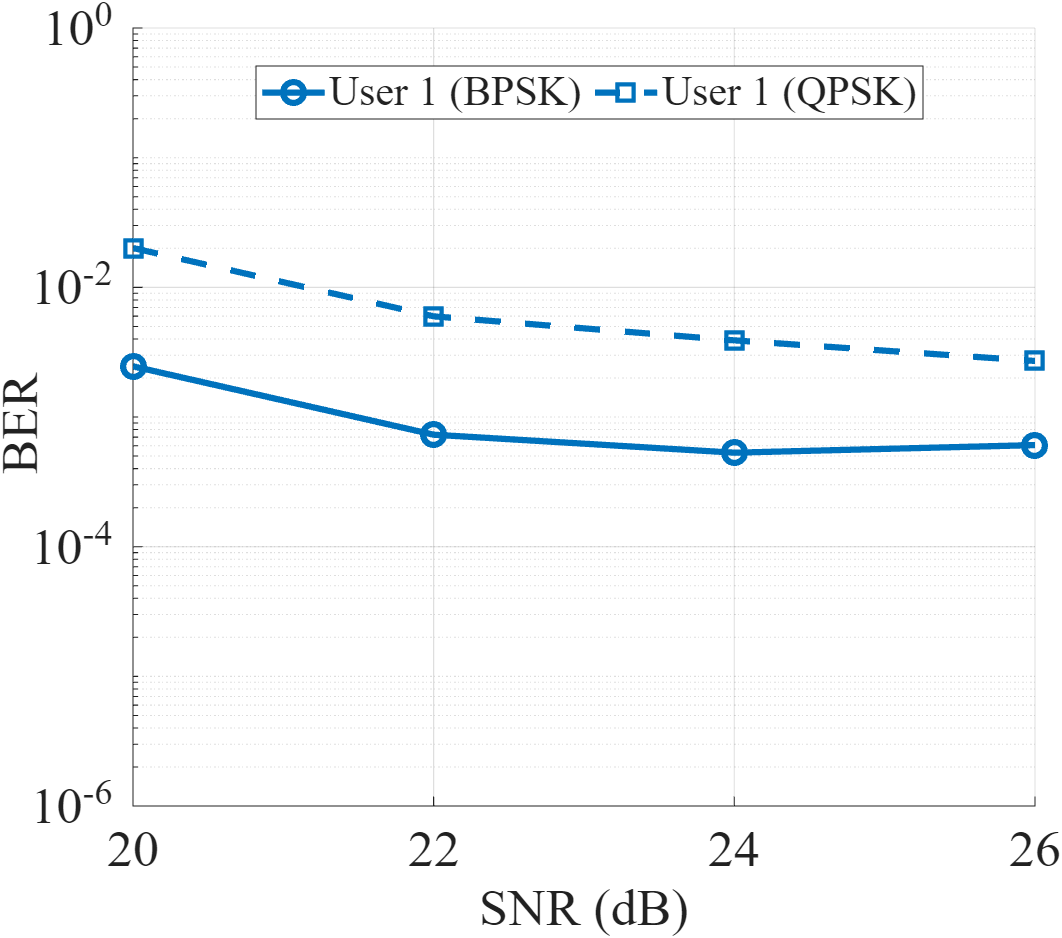}%
        \label{fig: Private  BPSK QPSK BER}%
    }
    \caption{BER under BPSK and QPSK signaling: (a) common data and (b) private message, with $P_{\text{max}}=1.5$ W.}
    \label{QPSK BER}
\end{figure}
\begin{figure}[t]
    \centering
    \subfloat[]{
        \includegraphics[width=0.23\textwidth]{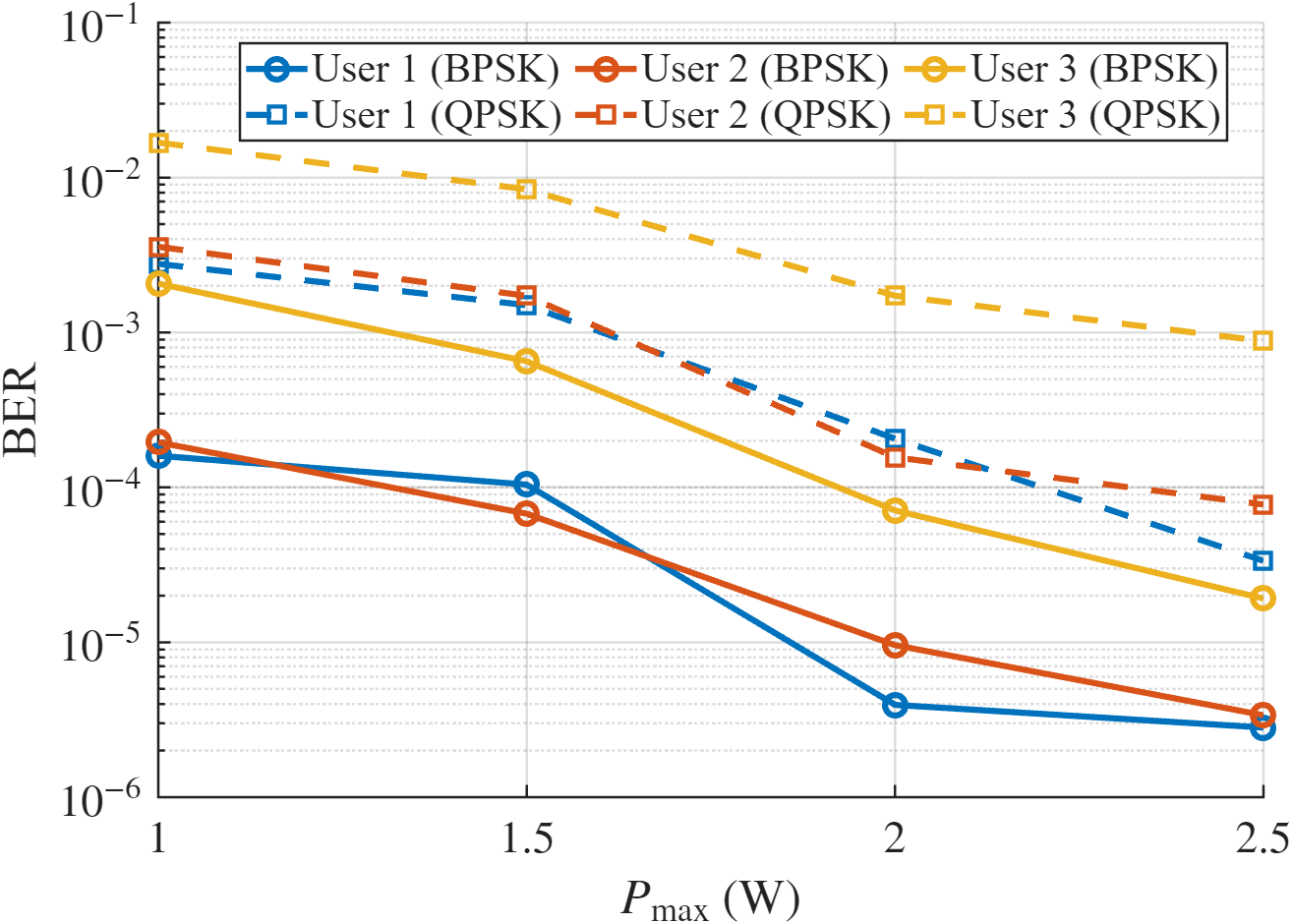}%
        \label{fig: Common powervar BER}%
    }\hfill
    \subfloat[]{
        \includegraphics[width=0.23\textwidth]{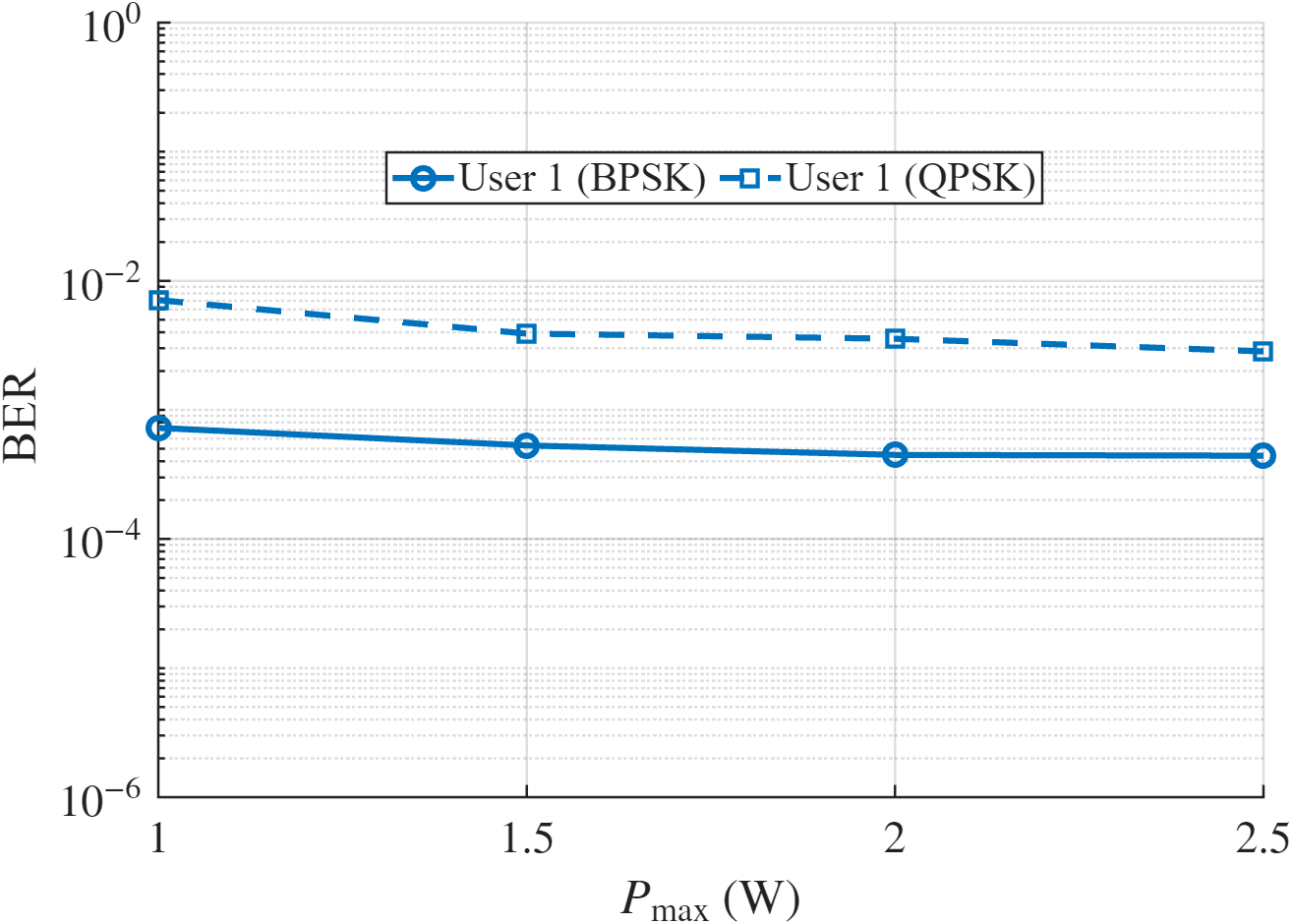}%
        \label{fig: private message powervar BER}%
    }
    \caption{BER  under BPSK and QPSK signaling versus $P_{\text{max}}$: (a) common data and (b) private message, with $\text{SNR}=24\,\text{dB}$.}
    \label{powervar_BER}
\end{figure}
\subsection{Channel Estimation and BER Performance}
Among the GS and GG configurations with favorable sum-SE performance, we focus
on GS for channel-estimation and BER evaluation because it provides a more
challenging setting due to interference from the superimposed private messages.
We use \textbf{L-2D} for resource allocation because it substantially reduces
computation time while achieving sum SE close to that of \textbf{SCA-2D}.
The minimum-rate requirement is set to $R^{\mathrm{th}}=1$~bps/Hz
in these experiments. 
Fig.~\ref{fig: channel_estmiatin} shows the channel normalized MSE (NMSE) of the proposed
LMMSE estimator for all users, where
$\text{NMSE}
=
\frac{\mathbb{E}[\|\pmb{h}_u-\hat{\pmb{h}}_u\|^2]}
{\mathbb{E}[\|\pmb{h}_u\|^2]}$.  
The NMSE remains low for all users, and the close agreement between the
theoretical and empirical results validates Theorem~1. 
We next examine the impact of channel-estimation errors on data detection
in Fig.~\ref{perfect_imperfect BER}.
As shown in Fig.~\ref{fig: Common Perfect_imperfect BER}, the common-data
BER is nearly identical under perfect and imperfect CSI.
For the private message, Fig.~\ref{fig: Private Perfect_imperfect BER}
shows an approximately \(3\)-dB loss for User~1 under imperfect CSI.
Users~2 and~3 are not shown because no private-message power is allocated
to them under the optimized power allocation.
Overall, these results show that the proposed channel estimator limits
the impact of CSI errors on subsequent data detection.

Fig.~\ref{QPSK BER} compares the BER performance of BPSK and QPSK.
As shown in Figs.~\ref{fig: Common BPSK QPSK BER} and
\ref{fig: Private BPSK QPSK BER}, QPSK exhibits higher BER than BPSK
for both the common and private messages.
This is because, for the same average symbol power, QPSK has a smaller
decision margin and is thus more sensitive to residual interference and SIC errors.
Fig.~\ref{powervar_BER} shows the BER performance for different total
power budgets.
As shown in Fig.~\ref{fig: Common powervar BER}, the common-data BER
decreases noticeably with \(P_{\max}\) for all users under both BPSK and
QPSK.
In contrast, Fig.~\ref{fig: private message powervar BER} shows only a
small improvement in the private-message BER of User~1.
This behavior follows from the optimized power allocation: since the
private messages act as interference when decoding the common message,
the optimizer limits their power and uses the additional power primarily
to improve common-message transmission.



\section{Conclusion}
{In this paper, we developed an OTFS-RSMA framework that jointly considers
pilot transmission, channel estimation, message detection, and resource
allocation under imperfect CSI.
Our results demonstrate the interplay between pilot and data transmission,
where the transmission structure affects both channel-estimation accuracy
and resource utilization.
In particular, guard-based common-message transmission provides favorable
sum-SE performance across the considered settings, while the preferred
private-message structure depends on the operating regime.
These results suggest that channel estimation and data transmission should
be considered jointly in the design of OTFS-RSMA systems.
We also showed that the resource-allocation problem admits a simple
allocation structure, enabling near-SCA performance with substantially
reduced computational complexity.
Future work may extend the proposed framework to fractional
Doppler, multi-antenna transmission, and SIC error propagation
caused by erroneous common-message decoding.}

\ifCLASSOPTIONcaptionsoff
  \newpage
\fi


 \vspace{-0.2cm}
\bibliography{bare_jrnl}
\bibliographystyle{IEEEtran}
%

%




\end{document}